%% file: main_camera_ready.tex
\documentclass[11pt]{article}

\usepackage[final]{acl}

\usepackage{times}
\usepackage{latexsym}

\usepackage[T1]{fontenc}

\usepackage[utf8]{inputenc}

\usepackage{microtype}

\usepackage{inconsolata}

\usepackage{graphicx}

\usepackage{subcaption}
\usepackage{booktabs} 
\usepackage{tabularx,array,threeparttable}
\newcolumntype{Y}{>{\raggedright\arraybackslash}X}
\usepackage{algpseudocode}
\usepackage{algorithm}
\usepackage{amsfonts}
\usepackage{amsmath}
\usepackage{amssymb}
\usepackage{mathtools}
\usepackage{amsthm}

\usepackage{url}
\usepackage{multirow}
\usepackage{bbm}
\usepackage{enumitem}
\usepackage{wrapfig}

\input{math_commands.tex}

\theoremstyle{plain}
\newtheorem{theorem}{Theorem}[section]

\newtheorem{lemma}[theorem]{Lemma}

\theoremstyle{definition}
\newtheorem{definition}[theorem]{Definition}

\theoremstyle{remark}

\title{Unsupervised Speech Recognition at the Syllable Level}
\author{Liming Wang$^1$\thanks{Now at the Chinese University of Hong Kong}, Kai-Wei Chang$^1$, Kunio Kashino$^2$, David Harwath$^3$,\\ 
\textbf{Mark Hasegawa-Johnson}$^4$, \textbf{James R. Glass}$^1$\\
$^1$Massachusetts Institute of Technology,
$^2$NTT, Inc.,\\
$^3$University of Texas at Austin,
$^4$University of Illinois Urbana-Champaign\\
\texttt{limingw@csail.mit.edu}
}

\newcommand\liming[1]{}
\newcommand{\dd}{\mathrm{d}}

\def\aishell/{AISHELL-3}
\def\hubl/{HuBERT}
\def\mincut/{min-cut}
\def\km/{K-means}
\def\reborn/{REBORN}
\def\sylcipher/{SylCipher}
\def\pdf/{probability distribution}
\def\prenet/{pre-net}
\def\postnet/{post-net}
\def\pyphen/{Pyphen}
\def\wtov/{wav2vec 2.0}
\def\wtovu/{wav2vec-U}

\def\subsecvsp{\vspace{-0.2cm}}
\def\figvsp{\vspace{-0.25cm}}

\begin{document}
\maketitle
\begin{abstract}
\input{sections/abs}
\liming{Emphasize on the semi-supervised contributions, relax the theoretical assumptions. A teaser plot+add multilingual features to figure 1}
\end{abstract}

\section{Introduction\liming{(READY)}}\label{sec:intro}
\liming{Phrase it to have more novelty}
\input{sections/intro}

\section{Related work\liming{(READY)}}\label{sec:related_work}
\input{sections/related_work}

\section{Syllable-level unsupervised speech recognition\liming{(READY)}}\label{sec:formulation}
\input{sections/problem_formulation}

\section{SylCipher: Syllable-level UASR via information-constrained masked language modeling\liming{(READY)}}\label{sec:method}
\input{sections/method}

\section{Experiments}\label{sec:exp}
\input{sections/experiments}

\section{Conclusion}\label{sec:concl}
\liming{Check format}
\input{sections/conclusion}

\section*{Acknowledgments}
We would like to thank NTT inc. for the funding support, and Junrui Ni for helping to set up his JSTTI code.

\section*{Limitations}\label{app:limit}
\sylcipher/ is not yet language-universal, since different languages use different writing systems and require linguistic knowledge to properly syllabify. For example, languages such as Hebrew and Arabic omits vowels in their writing system, which could pose challenges to existing syllabifiers.
While our experiments show that the method can be adapted to more resource-efficient tokenizers such as BPE with minimal modifications, coming up with a language-universal tokenization method remains an open problem.
Further, the iterative training procedure can be further simplified into an end-to-end approach. Lastly, improving the robustness of \sylcipher/ under domain mismatch between speech and text remains an open challenge. 
\paragraph{Use of AI assistants.} We used AI assistants only for language polishing and improving clarity. All technical content, claims, experiments, and conclusions were developed and verified by the authors.


\bibliography{emnlp2026}

\appendix

\input{sections/appendix}
\label{sec:appendix}

\end{document}

%% file: math_commands.tex
\usepackage{amsmath,amsfonts,bm}

\def\eqref#1{equation~\ref{#1}}

\def\1{\bm{1}}

\newcommand{\brac}[1]{\left(#1\right)}
\newcommand{\sbrac}[1]{\left[#1\right]}

\DeclareMathAlphabet{\mathsfit}{\encodingdefault}{\sfdefault}{m}{sl}
\SetMathAlphabet{\mathsfit}{bold}{\encodingdefault}{\sfdefault}{bx}{n}

\def\tX{{\tilde{X}}}

\def\gX{{\mathcal{X}}}
\def\gY{{\mathcal{Y}}}

\def\sR{{\mathbb{R}}}

\newcommand{\E}{\mathbb{E}}
\newcommand{\Ls}{\mathcal{L}}

\newcommand{\KL}{D_{\mathrm{KL}}}

\DeclareMathOperator*{\argmax}{arg\,max}

%% file: sections/abs.tex
Training speech recognizers with unpaired speech and text -- known as unsupervised speech recognition (UASR) -- is a crucial step toward extending ASR to low-resource languages in the long-tail distribution and enabling multimodal learning from non-parallel data. However, existing approaches based on phones often rely on costly resources such as grapheme-to-phoneme converters (G2Ps) and struggle to generalize to languages with ambiguous phoneme boundaries due to training instability. In this paper, we address both challenges by introducing a syllable-level UASR framework based on masked language modeling, which avoids the need for G2P and the instability of GAN-based methods. Our approach achieves up to a 40\% relative reduction in character error rate (CER) on LibriSpeech and generalizes effectively to low-resource languages that have remained particularly difficult for prior methods. Code is publicly available\footnote{https://github.com/cactuswiththoughts/SylCipher}.\liming{create a github}

%% file: sections/intro.tex
Recent advances in self-supervised learning~\citep{Baevski2020-wav2vec2,Hsu2022-hubert,chen2022-wavlm,chung2021w2vbert,chen2024-xeus,mohamed2022self} and spoken language modeling~\citep{arora2025-slm,chu2024qwen2,defossez2024moshi} have made voice assistants increasingly capable. Yet these systems remain far from \emph{language-universal}: most support only a handful of world's languages~\citep{chen2024-xeus,conneau21_interspeech}, largely because large paired speech-text corpora are unavailable for many languages. 
\begin{figure}
    \centering
    \includegraphics[width=0.49\textwidth]{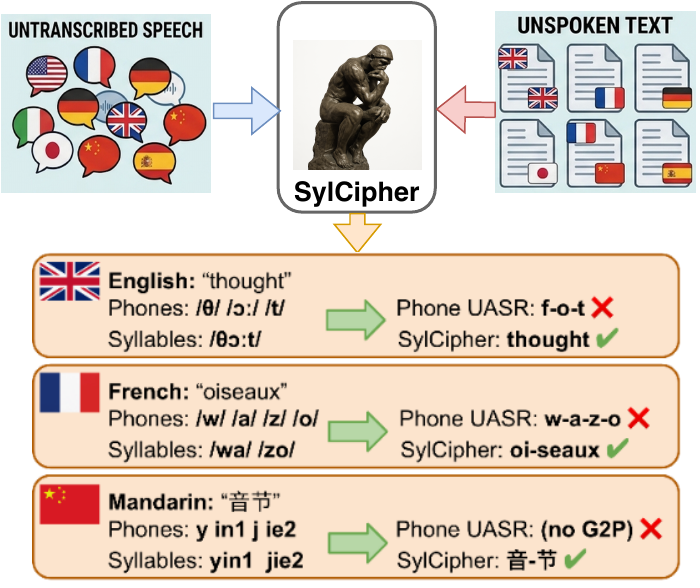}
    \caption{Syllable-level speech-text decipherment.}
    \label{fig:teaser}
\end{figure}
\begin{figure*}
    \centering
    \begin{subfigure}{0.5\textwidth}
        \centering
        \includegraphics[width=0.95\textwidth]{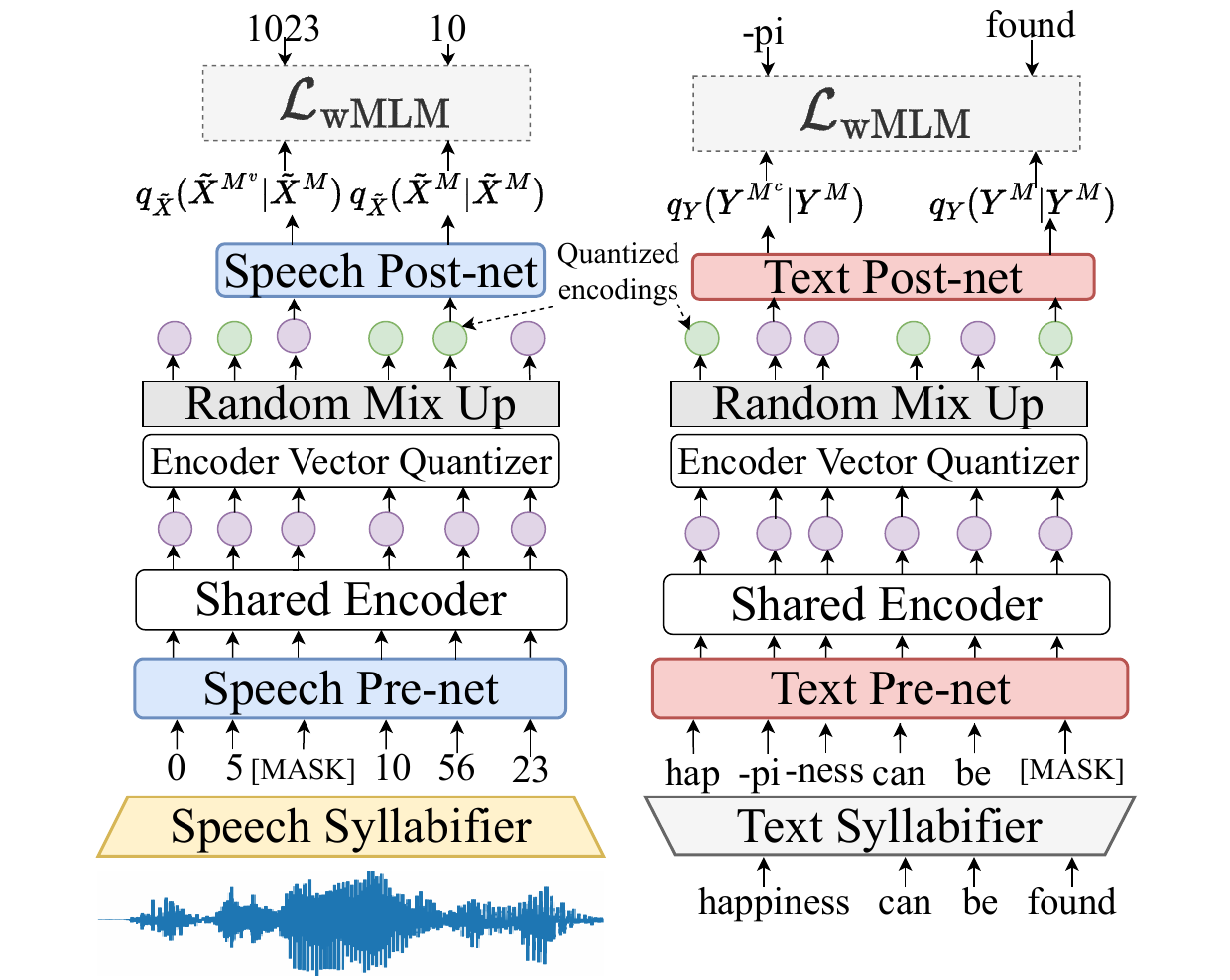}
        \caption{MLM-based stages}
    \end{subfigure}
    \begin{subfigure}{0.23\textwidth}
        \centering
        \includegraphics[width=0.94\textwidth]{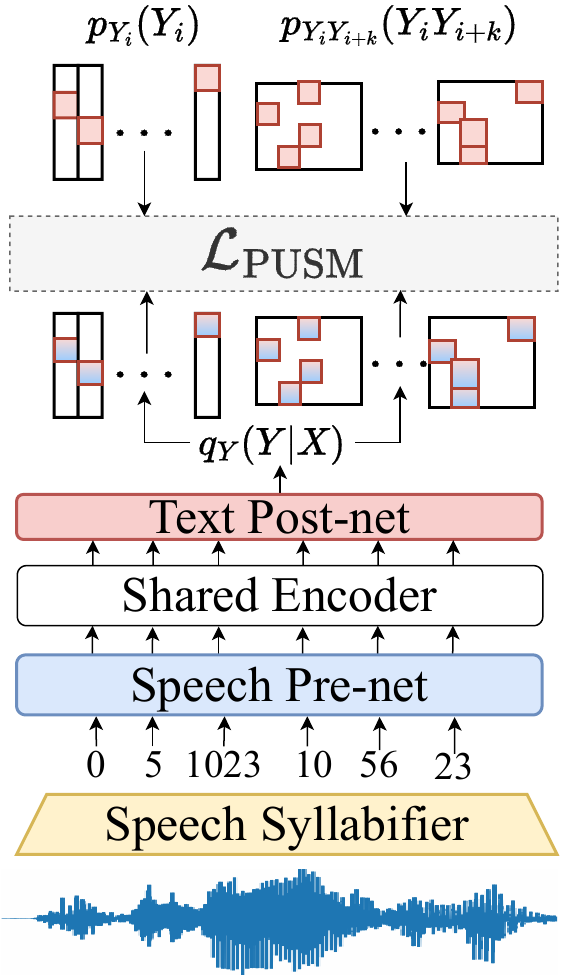}
        \caption{PUSM stage}
    \end{subfigure}
    \begin{subfigure}{0.24\textwidth}
    \centering
    \includegraphics[width=0.95\textwidth]{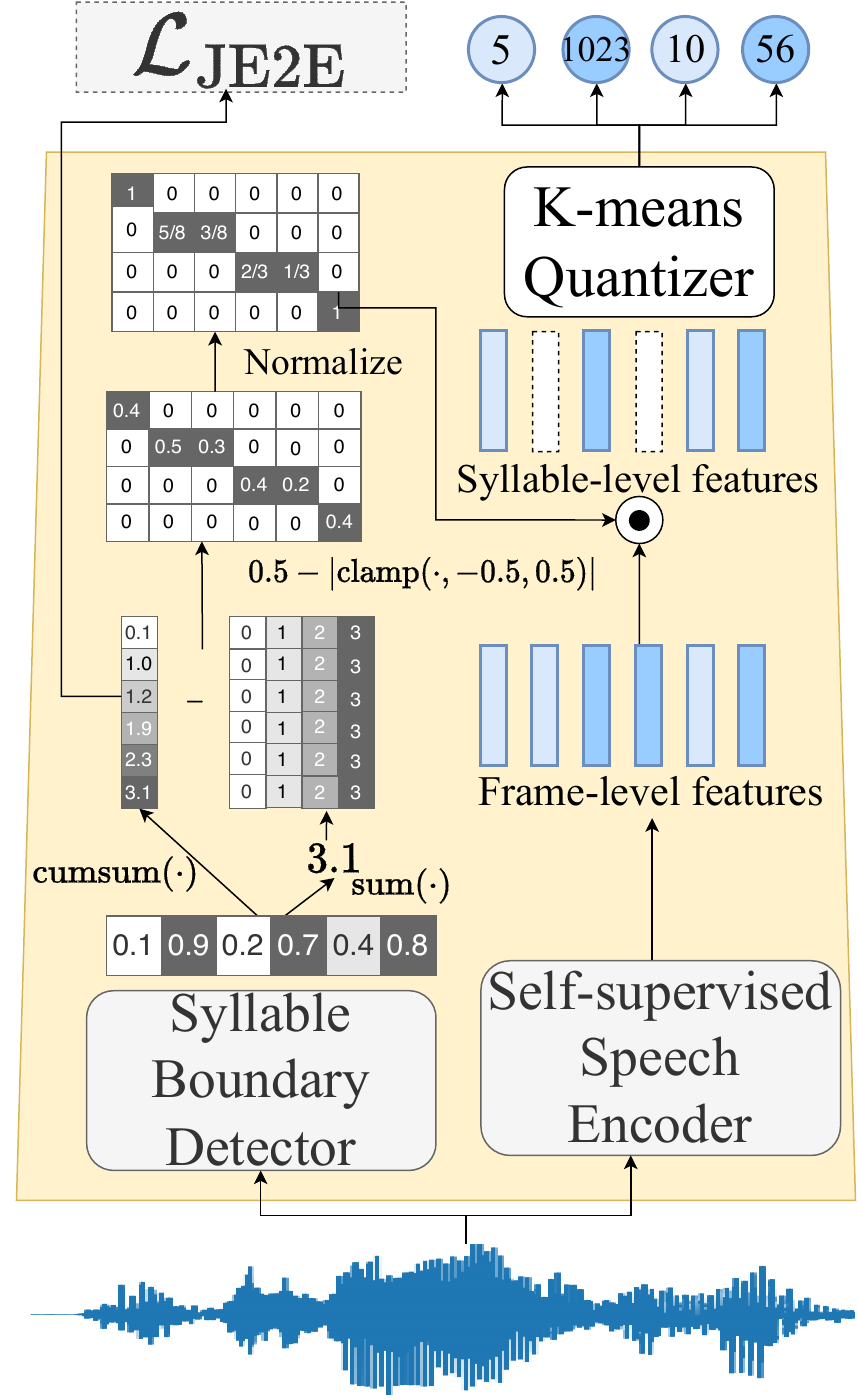}
    \caption{Speech syllabifier}
    \label{fig:sylcipher_c}    
    \end{subfigure}
    \caption{\textbf{Overall architecture of \sylcipher/}. Gray boxes are fixed during training. (a) MLM-based stages: learn a compressed joint semantic space with a shared encoder and random mix-up. (b) PUSM stage: align speech and text spaces by matching lower-order marginals of their distributions.} 
    \label{fig:sylcipher}
    \figvsp
    \figvsp
\end{figure*}

A natural step toward language-universal speech technology is to build recognizers from \emph{unpaired} speech and text, i.e., \emph{unsupervised speech recognition} (UASR)~\citep{Glass2012-unsup-speech,Baevski2021-wav2vec-u,Liu2023-wav2vecu2,Tseng2024-reborn}. 
UASR is important not only for speech recognition itself, but also for downstream applications such as speech synthesis~\citep{Ni-unsuptts-interspeech2022,Liu2022-utts}, translation~\citep{wang-etal-2023-simple}, and spoken language understanding~\citep{shi2023-unsupslu}. More broadly, it is a core instance of \emph{unpaired multimodal learning}~\citep{Artetxe18-unsupnmt,artetxe2018unsupervised,artetxe-etal-2019-effective,Lample2018-unsupmtmono,lample2018phrase,ma2019-unpaired,hoshen2018non,goldwasser2023theory,levy2025unsupervised}, where modalities need to be aligned without parallel data. In the absence of sentence-level alignment, the model must infer higher-level linguistic units--phones, syllables, and words--from raw speech waveforms in conjunction with global text statistics. Thus, UASR provides broader insights into representation learning and multimodal alignment without supervision.

Most strong UASR systems~\citep{Baevski2021-wav2vec-u,Liu2023-wav2vecu2,Tseng2024-reborn} operate at the \emph{phoneme-level}. This requires converting written symbols (graphemes) into phonemes using a grapheme-to-phoneme convertor (G2P). In many languages, however, G2P is difficult or unavailable because writing systems may omit key phonetic information, vary substantially across regions, or fail to represent important sound contrasts~\cite{KandybowiczTorrence2017,FarukhVulchanova2014,KhanEtAl2024PashtoNLP}. Even when pronunciation dictionaries exist, building and maintaining them is labor-intensive. Without a G2P, phoneme-based UASR often degrades sharply because speech and raw text become poorly aligned~\cite{Liu2023-wav2vecu2,Ni-unsuptts-interspeech2022}. In addition, phone boundaries can be hard to detect in languages with strong co-articulation, such as Mandarin.

A word-level alternative avoids G2P, but introduces a different problem: the vocabulary of rare words is effectively unbounded, making coverage and generalization much harder. Word segmentation also depends on longer-range context, which can destabilize the segmentation mechanisms used in current UASR systems~\cite{wang-etal-2023-unsup-speech2sign}.

In this work, we instead study UASR at the \emph{syllable-level}. This choice has three advantages. First, unlike words, the number of syllables is finite, which reduces the long-tail token problem. Second, in many languages, speech and text align more naturally at the syllable level than at the phone or word level. 
Mandarin is a clear example: written characters often correspond closely to spoken syllables, making syllables a more natural unit than phonemes. 
Third, recent progress in syllable boundary detection and unit discovery~\citep{Baade2025-syllablelm,Cho2025-sylber} makes it possible to segment speech into syllable-like units without text supervision, often more reliably than unsupervised word segmentation methods while retaining much of the relevant linguistic structure~\cite{peng2022-vghubert,fuchs2023-gradseg}.

In this paper, we make the following contributions.
\begin{enumerate}[leftmargin=2em, labelsep=0.5em, itemindent=0.0em]
    \item We introduce \textbf{\sylcipher/}, to our knowledge the \textbf{first syllable-based UASR system}. \sylcipher/ jointly predicts syllable boundaries and discrete speech units from raw speech under a unified self-supervised objective, avoiding adversarial training and improving training stability. 
    \item We evaluate \sylcipher/ across domains, languages and level of resources. On LibriSpeech, it improves over prior G2P-free UASR systems by up to 40\% relative character error rate (CER). On SpokenCOCO, the gains are even larger, demonstrating robustness across domains. On Multilingual LibriSpeech,  we build the \textbf{first multilingual UASR system} to our knowledge by training \sylcipher/ on multilingual data and outperform prior G2P-free, monolingual approaches by up to 50\% relative CER. On Mandarin, it achieves 12.2\% phone error rate (PER), outperforming GAN-based UASR methods that fail to even converge. On two low-resource languages, it consistently outperforms phoneme-level approaches.
    \item We also apply \sylcipher/ to the \textbf{semi-supervised setting} and demonstrate that pseudo-labels from our system can improve semi-supervised ASR performance across multiple languages with up to 20 hours of paired data.
    \item We perform careful ablation and error analysis, examining the effects of token vocabulary size, syllabifier choice, and segmentation mechanisms. In particular, we  evaluate \textbf{BPE units as a text-derived, noisier approximation} to syllable-like segmental units. These experiments show that the proposed framework can be extended beyond exact syllable inventories, while also highlighting that linguistically or acoustically motivated syllable units provide a stronger inductive bias for UASR.
\end{enumerate}
\textbf{Paper organizations.} Section~\ref{sec:related_work} reviews related work. Section~\ref{sec:formulation} formalizes the syllable-level UASR problem. Section~\ref{sec:method} presents \sylcipher/ and its theoretical guarantees. Section~\ref{sec:exp} reports experiments and ablations. Section~\ref{sec:concl} concludes with limitations and future work. 

%% file: sections/related_work.tex
\paragraph{Unsupervised speech recognition} Early work on UASR assumed the existence of a reliable G2P and formulate the problem as an adversarial game, where a conditional generator predicts a phoneme sequence given a speech waveform, and a discriminator tries to tell real phonemized text apart from the generator's output~\citep{wang-etal-2023-unsupasr-theory}. Within this framework, several works have explored different segmentation mechanisms such as fixed unsupervised phoneme segmentation~\citep{Liu2018-asru}, iterative forced alignment~\citep{Chen2019-uasr}, de-duplication~\cite{Baevski2021-wav2vec-u} and reinforcement learning~\citep{Tseng2024-reborn}. \citep{wang-etal-2023-unsup-speech2sign} proposed an alternative formulation of UASR as an explicit distribution matching problem, by matching the lower-order $N$-gram distributions of the generated and real texts. Later work further extended this reformulation to use G2P-free tokens such as words~\citep{wang-etal-2023-unsup-speech2sign,Ni2025-jstti}, more expressive architecture such as transformers~\citep{Ni2025-jstti} and more general text distribution such as masked prediction probabilities~\citep{Ni2025-jstti}. We adapt this word-level UASR approach to the syllable-level, and supplement it with a simplified training flow, a more stable differentiable boundary detector and additional learning objectives.
\paragraph{Syllable-level self-supervised learning} Early work on syllable-level modeling of speech used signal processing techniques~\citep{Zhang2009-syllable}. To discover higher-level structure and more efficient self-supervised learning (SSL) representations from raw speech, \citep{Peng2023-syllable,cho2023-sdhubert} proposed to induce syllabic structure from existing SSL models such as HuBERT~\citep{Hsu2022-hubert}. To this end,~\citep{Peng2023-syllable} probed the self-attention layers of VG-HuBERT~\citep{peng2022-vghubert}, a visually grounded SSL model to detect syllable-like feature clusters and further refine such clusters using a \mincut/ algorithm~\cite{shi1997normalized}. To sidestep the need for visual data, \citep{cho2023-sdhubert} employed a speech-only SSL model trained with utterance-level self-distillation from a HuBERT teacher. Due to the indirect manner of such approaches by which syllabic structures are derived, they are often noisy and unreliable. To cope with this issue, recent works~\citep{Baade2025-syllablelm,Cho2025-sylber} proposed a more direct and targeted approach by performing self-distillation at the syllable-level, which significantly improved syllable boundary detection and unit discovery performance by encouraging sharper contrast between within and between-syllable feature frames. 

%% file: sections/problem_formulation.tex
In this section, we formulate syllable-level UASR. 
Let $X=[X_1,\cdots,X_T]\in\gX^T$ be a padded sequence of speech feature vectors and $Y=[Y_1,\cdots,Y_L]\in\gY^L$ be a padded sequence of text tokens in the same language. Since tokenized speech is typically longer than its text transcription, we assume $T\geq L$. We further assume $X$ and $Y$ come from two \emph{unpaired} datasets and are therefore statistically independent. Meanwhile, the two modalities are assumed to be \emph{matched}: there exists an ASR function $y^*:\gX\mapsto\gY$ such that their distributions $p_X$ and $p_Y$ satisfy
\begin{align}
    p_Y(y)=\int_{x\in\gX:y^*(x)=y} p_X(x)\dd x,\,\forall y\in\gY^L. \label{eq:matched}
\end{align}
The goal of UASR is to recover $y^*$ given only unpaired $X$ and $Y$. In the syllable-level setting, the tokens are syllables. The problem is closely related \emph{decipherment}, where one decodes a message written in an unknown script without a lexicon or grammar. 

In practice, $X$ is formed from frame-level SSL features~\cite{Hsu2022-hubert}, and $p_X$ and $p_Y$ are only approximately matched because of finite-sample noise and domain mismatch. Moreover, UASR is ill-posed in general. However,~\citep{wang-etal-2023-unsupasr-theory} show that the mapping $y^*$ in \eqref{eq:matched} becomes identifiable when syllable boundaries are known and the language satisfies mild conditions.

%% file: sections/method.tex
In this section, we present \textbf{SylCipher}, a syllable-level framework for UASR. The model first converts continuous speech into syllable-like units, then learns a shared representation space in which speech and text are processed by the same encoder. We first describe the architecture and training objective, then give the theoretical motivation, and finally introduce the practical training modifications.

At a high level, \sylcipher/ has three goals:
\begin{itemize}[leftmargin=2em, labelsep=0.5em, itemindent=0.0em]
    \item Convert speech into syllable-level units that are easier to align with text;
    \item Learn a shared latent space for speech and text with a single encoder;
    \item Make UASR training stable and effective in practice.
\end{itemize}

\subsecvsp
\subsection{Training: UASR via information compression}
\subsecvsp
As shown in Figure~\ref{fig:sylcipher}, SylCipher is an encoder-only language model with a \emph{shared encoder} for speech and text. Two uni-modal, linear \emph{\prenet/s} map the two modalities into a common embedding space: $e_{\tilde{X}}:\gX^T\mapsto\sR^{L\times d}$ for speech and $e_Y:\gY^L\mapsto\sR^{L\times d}$ for text. 

Before the speech \prenet/, a \emph{speech syllabifier} converts frame-level acoustic features into a syllable-level sequence. It contains (i) a \emph{differentiable soft-pooler} $m:\gX^T\mapsto\gX^L$ that aggregates frames into syllable segments and (ii) a tokenizer $c:\gX^L\mapsto \tilde{\gX}^L$ that discretizes the pooled representations. The speech pipeline is therefore 
\begin{align}
    e_{\tilde{X}}(\tilde{X}) = e_{\tilde{X}}\circ c\circ m(X),
\end{align}
where $\tilde{X}:=c\circ m(X)\in\tilde{\gX}^L.$

The soft-pooler first employs a trainable boundary detector to predict frame-level boundary probabilities $b(X)\in[0, 1]^T$. The boundary detector is pretrained with initial boundaries $[\hat{b}_1(X),\cdots,\hat{b}_T(X)]$ from the unsupervised syllable detector Sylber-\citep{Cho2025-sylber}. These frame-level boundary probabilities are converted into a soft pooling mask $a(X)$, which maps frame-level speech features into segment-level representations: 
\begin{align}\label{eq:softpooler}
    a_{it}(X) &:= \sigma_{\epsilon}\brac{i-\sum_{\tau\leq t}b_{\tau}(X)},\\
    m_i(X) &:= \sum_{t=1}^T \frac{a_{it}(X)}{\sum_{\tau=1}^T a_{i\tau}(X)}X_t, 
\end{align}
where $\sigma_{\epsilon}(x):=\epsilon - |\mathrm{clamp}(x,-\epsilon,\epsilon)|$. 
This produces sparse pooling weights and avoids the numerical instability of earlier soft-pooling methods~\citep{Bhati2022-scpc,wang2024unsupervised}. The segment-level representations are passed to a vector quantizer and speech \prenet/.

Both modalities are then processed by the same encoder $f$:
\begin{align}
    f_{\tilde{X}}(\tilde{X}) := f\circ e_{\tilde{X}}(\tilde{X}),\quad f_Y(Y) := f\circ e_Y(Y).
\end{align}
In practice, $c$ is implemented with a differentiable \km/ vector quantizer~\citep{Ni2025-jstti}, and $f$ is a multi-layer transformer~\citep{Vaswani2017}. 

\paragraph{Distribution matching.} 
Without paired speech-text data, we train the model by matching the unimodal speech and text distributions. Let $g_{\tilde{X}}$ and $g_Y$ be speech and text \emph{\postnet/s} built on top of the shared encoder, with
\begin{align}\label{eq:postnet}
    q_{Z}(Z) &:= g_Z\circ f_Z(Z)=\prod_{i=1}^L g_{Z,i}(Z_i|f_Z(Z_{1:(i-1)}))\nonumber\\
    &=\prod_{i=1}^{L}\frac{e^{W^Z_{iZ_i}f_Z(Z_{1:(i-1)})}}{\sum_{z}e^{W^Z_{iz}f_Z(Z_{1:(i-1)})}},
\end{align}
where $Z\in \{X, Y\}$, $Z_i$ is the syllable unit at segment $i$ and $g_{Z,i}$ is the linear-softmax post-net prediction head, following JSTTI~\cite{Ni2025-jstti}. Thus, $q_Z$ defines a sequence-level distribution computed as the product of the softmax probabilities.

Training minimizes \emph{Kullback–Leibler (KL) divergence} to the two unimodal distributions:
\begin{align}\label{eq:regularized_distribution_matching}    \min_{q_{\tilde{X}},q_Y}\KL(p_{\tilde{X}}||q_{\tilde{X}})+\KL(p_Y||q_Y),
\end{align}
while \textbf{restricting the capacity} of the shared representation to prevent speech and text from drifting into disjoint regions of the latent space. We prove that Eq.~\ref{eq:regularized_distribution_matching} is able to match the true and generated text distributions under regularity conditions~\citep{wang-etal-2023-unsupasr-theory}, and defer the formal objective, assumptions, and theoretical guarantee to Appendix~\ref{app:proof_of_main} and ~\ref{app:assumption_discussion}.

\paragraph{Practical training.} In practice, \sylcipher/ uses three training components. First, 
it applies weighted masked language modeling (wMLM) on both speech and text to learn the shared encoder and \postnet/s to approximate the KL terms in Eq.~\ref{eq:regularized_distribution_matching}. Given a mask distribution $p_M$, we minimize the weighted loss:
\begin{multline*}\label{eq:mlm}
    \Ls_{\mathrm{wMLM}}(\theta_X,\theta_Y):=\\
    -\E_{Z\sim \frac{1}{2}p_{\tilde{X}}+\frac{1}{2}p_Y,\,M\sim p_M}\sbrac{\ln q_Z(Z^M|Z^{M^c})}\\
    -\lambda \E_{Z\sim \frac{1}{2}p_{\tilde{X}}+\frac{1}{2}p_Y,\,M\sim p_M}[\ln q_Z(Z^M|Z^M)],
\end{multline*}
where $Z^M$ is masked portion of the embedding, $Z^{M^c}$ is the unmasked portion of the embedding, $\theta_X$ and $\theta_Y$ are speech-related and text-related parameters, and $\lambda$
balances the function of the decoder for mask prediction and reconstruction. Note that for $\lambda=0$, Eq.~\ref{eq:mlm} approximates Eq.~\ref{eq:regularized_distribution_matching} since
\begin{align*}
&\KL(p_{Z}||q_{Z})\propto -\E_{p_{Z}}\ln q_{Z}(Z)\\
\propto&-\E_{p_Z,p_M}\ln q_Z(Z^M|Z^{M^c})=\Ls_{\mathrm{wMLM}},
\end{align*}
where $p_Z$ and $q_Z$ denote the true and estimated unimodal distributions respectively. 

Second, since the computations Eq.~\ref{eq:softpooler} is differentiable, the boundary detector is updated jointly with the UASR objectives, allowing joint end-to-end (JE2E) training to refine the boundaries using text-distribution matching signals. To discourage both over- and under-segmentation, a soft syllable-count constraint. 
\begin{align*}
    \Ls_{\mathrm{JE2E}}(\theta_X) := \E_{X}\left|\sum_{t=1}^T b_t(X)-\hat{b}_t(X)\right|.
\end{align*}
Third, once MLM saturates, it adds positional unigram and skipgram matching (PUSM)~\citep{wang2024unsupervised,Ni2025-jstti} for explicit distribution matching.

The overall training objective is
\begin{align}
    \Ls := 
    \lambda_1\Ls_{\mathrm{wMLM}}+\lambda_2\Ls_{\mathrm{JE2E}}+\lambda_3\Ls_{\mathrm{PUSM}}.
\end{align}
Because the PUSM stage requires different batch sizes, training follows a simple iterative schedule: (1) train with fixed boundaries and wMLM only. (2) enable JE2E to refine segmentation, and (3) switch to PUSM for explicit distribution matching. To restrict the capacity of the shared representation during training, we use a shallow encoder and partial quantization of encoder states before the \postnet/s with random mix-up~\citep{Ao2022-speecht5}.

\subsecvsp
\subsection{Inference}
\subsecvsp
During inference, the ASR system cascades the speech syllabifier, the speech pre-net, and the text post-net:
\begin{align}
    y(X)&:=\argmax_{Y\in\gY^L}\; q_Y(Y|X)\nonumber\\
    &:=\argmax_{Y\in \gY^L}\;(g_Y\circ f_X(X))(Y).
\end{align}
We also find that replacing each \texttt{<OOV>} token with the second most likely prediction reduces CER by about $1\%$ relative to simply discarding \texttt{<OOV>}s.

%% file: sections/experiments.tex
We first describe the datasets and preprocessing, then present the main UASR, boundary-detection, semi-supervised, and ablation results. Additional implementation details are deferred to Appendix~\ref{app:implementation}.

\subsecvsp
\subsection{Datasets}\label{sec:dataset}
We evaluate \sylcipher/ on six datasets: LibriSpeech (460 hours of clean read English)~\citep{Panayotov15-LibriSpeech}, SpokenCOCO (742 hours of spoken image captions)~\citep{hsu2021textfree}, Mandarin \aishell/ (85 hours)~\citep{shi21c_interspeech}, multilingual LibriSpeech (MLS) for German, Dutch and French~\citep{Pratap2020MLSAL}, and a low-resource Cantonese-Taigi setting built from MDCC~\citep{yu2022automatic} and SuiSiann~\citep{ithuan2019suisiann}. For MLS, we subsample 100 hours per language following prior work~\citep{Tseng2024-reborn}, append a language ID token to each tokens as well as the start of each speech and text unit sequences, train the MLM stage jointly, and run the PUSM stage monolingually.  

We follow the LibriSpeech split of \citep{Ni2025-jstti}, and the standard splits elsewhere. For LibriSpeech and SpokenCOCO, we report both a \emph{matched} setting, where empirical speech and text \pdf/s can be matched exactly, and an \emph{unmatched} setting, where they cannot. 
In the matched case, pairings are removed from paired corpora; in the unmatched case, LibriSpeech uses LibriLM~\citep{panayotov2015librispeech} with overlapping text removed, while SpokenCOCO is split into speech-only and text-only halves. For the remaining datasets, we report only the matched setting. We apply a voice activity detector\footnote{https://github.com/wiseman/py-webrtcvad.git} to all datasets.

\subsecvsp
\subsection{Speech and text syllabification}\label{sec:preproc}
For English text, we use \pyphen/\footnote{https://github.com/Kozea/Pyphen}, a rule-based hyphenation tool repurposed for syllabification without a G2P, together with a simple fallback for long words that remain unsplit (Appendix~\ref{app:code_pyphen+}); We refer to this combination as \pyphen/+. We also test other G2P-free approaches such as byte-pair encoding (BPE)~\cite{liu-etal-2025-superbpe,sennrich-etal-2016-neural}, and find \sylcipher/ robust to syllabification noise. For English, we keep only the top-2048 most frequent English syllables and replace the rest with a special \texttt{<OOV>} token (replacing $7\%$ of tokens in LibriSpeech and $2\%$ in SpokenCOCO). Similar syllabification procedure is used for MLS but with top-1024 syllables per language. For Mandarin, we use the pinyin or raw character as syllables. We keep the top-1024 types, covering 99.5\% of occurrences in the case of pinyin and 92\% in the case of raw character. Similar procedure is used for Cantonese and Taigi. For speech, we cluster syllable-level features obtained by mean pooling within Syllable boundaries, using a codebook size equal to the number of non-\texttt{<OOV>} text tokens.
\begin{table}[t]
    \begin{subtable}{0.49\textwidth}
    \centering
    \resizebox{0.95\textwidth}{!}{
    \begin{tabular}{llcccc}
    \toprule
        \multirow{2}{*}{\textbf{Model}} &  \multirow{2}{*}{\textbf{Student}} & \multirow{2}{*}{\textbf{Token}} & \multicolumn{1}{c}{\textbf{Matched}} & \multicolumn{1}{c}{\textbf{Unmatched}} \\
        & & & \textbf{CER ($\downarrow$)} 
        & \textbf{CER ($\downarrow$)} 
        \\
    \midrule
    \midrule   
    \multicolumn{5}{c}{\emph{G2P-based approach}} \\
    \midrule
    \multirow{1}{*}{\wtovu/*~\citep{Baevski2021-wav2vec-u}} & No & Phone & - & 13.3\\
    \multirow{1}{*}{\wtovu/ 2.0*~\citep{Liu2023-wav2vecu2}} & No & Phone & - & 12.2\\
    \multirow{1}{*}{\reborn/*~\citep{Tseng2024-reborn}} & No & Phone & - & 8.3 \\
    \midrule
    \multicolumn{5}{c}{\emph{G2P-free approach}} \\
    \midrule
        \multirow{2}{*}{\wtovu/~\citep{Baevski2021-wav2vec-u}}  & No & Char. & 35.6 & 43.3 \\
        & Yes & Char. & 33.8 & 42.1 \\
        \multirow{1}{*}{\reborn/~\citep{Tseng2024-reborn}} & No & Char. & 37.8 & 76.6\\
        JSTTI (forced align) & No & Char & 81.4 & 81.1\\
        \multirow{1}{*}{JSTTI~\citep{Ni2025-jstti}} & No & Word & 49.5 & 54.2\\
        \multirow{2}{*}{PUSM~\cite{wang-etal-2023-unsup-speech2sign}} & No & Syllable & 35.5 & 57.7 \\
        & Yes & Syllable & 33.0 & 54.7\\
        \midrule
        \sylcipher/ (Ours, forced align) & No & Syllable & 38.5 & 46.4 \\
        \sylcipher/ (Ours, Sylber) & No & Syllable & 43.5 & 48.6\\
        \sylcipher/ (Ours, Sylber+JE2E) & No & Syllable & 39.2 & 46.8 \\
        \multirow{2}{*}{\sylcipher/ (Ours, Sylber+JE2E+PUSM)} & No & Syllable & \textbf{21.8} & \textbf{35.9} & \\
        & Yes & Syllable & \textbf{17.5} & \textbf{33.3}\\
    \bottomrule
    \end{tabular}}
    \caption{UASR results on LibriSpeech (clean subsets)}
    \label{tab:uasr_libri}
    \end{subtable}
    \begin{subtable}{0.49\textwidth}
    \centering
    \resizebox{0.9\textwidth}{!}{
    \begin{tabular}{llcccc}
    \toprule
        \multirow{2}{*}{\textbf{Model}} &  \multirow{2}{*}{\textbf{Student}} & \multirow{2}{*}{\textbf{Token}} & \multicolumn{1}{c}{\textbf{Matched}} & \multicolumn{1}{c}{\textbf{Unmatched}} \\
        & & & \textbf{CER ($\downarrow$)} 
        & \textbf{CER ($\downarrow$)} 
        \\
    \midrule
    \midrule
        \multirow{2}{*}{\wtovu/~\citep{Baevski2021-wav2vec-u}}  & No & Char. & 45.0 & 45.2 \\
        & Yes & Char. & 46.3 & 35.3\\
        
        JSTTI & No & Char. & 78.3 & 100 \\
        \multirow{1}{*}{JSTTI~\citep{Ni2025-jstti}} & No & Word & 64.5 & 64.5\\
        \multirow{2}{*}{PUSM~\cite{wang-etal-2023-unsup-speech2sign}} & No & Syllable & 41.5 & 41.3 \\
        & Yes & Syllable & 34.7 & 34.3 \\
        \midrule
        \sylcipher/ (Ours, Sylber) & No & Syllable & 34.9 & 36.1 \\
        \sylcipher/ (Ours, Sylber+JE2E) & No & Syllable & 31.2 & 32.4 \\
        \multirow{2}{*}{\sylcipher/ (Ours, Sylber+JE2E+PUSM)} & No & Syllable & \textbf{23.4} & \textbf{26.8} \\
        & Yes & Syllable & \textbf{13.9} & \textbf{17.6}\\
    \bottomrule
    \end{tabular}}
    \caption{UASR results on SpokenCOCO}
    \label{tab:uasr_scoco}
    \end{subtable}
    \figvsp
    \caption{\textbf{UASR results on LibriSpeech (clean subsets) and SpokenCOCO}. Inside the bracket lists the unsupervised boundary used for each model. \emph{Student} stands for the student model used during the self-training stage using pseudo-labels from each model. For tokens used for the text data, \emph{Char.} stands for characters and \emph{Syllable} stands for syllable-level tokens converted using the Pyphen+ syllabifier without using a G2P. CER stands for character error rate. * indicates evaluation on LibriSpeech dev-clean instead.}
\end{table}

\begin{table}[t]
    \centering
    \resizebox{0.49\textwidth}{!}{
    \begin{tabular}{llcccc}
    \toprule
        \multirow{2}{*}{\textbf{Model}} &  \multirow{2}{*}{\textbf{Student}} & \multirow{2}{*}{\textbf{Token}} & \multicolumn{1}{c}{\textbf{w/o Tone}} & \multicolumn{1}{c}{\textbf{w. Tone}} \\
        & & & \textbf{PER ($\downarrow$)} 
        & \textbf{PER ($\downarrow$)} 
        \\
    \midrule
    \midrule
        \multirow{1}{*}{\wtovu/~\citep{Baevski2021-wav2vec-u}}  & No & Phone & 74.9 & 76.2 \\
        JSTTI & No & Init./Final & 96.4 & 100\\
        \multirow{1}{*}{JSTTI~\citep{Ni2025-jstti}} & No & Word & 83.2 & 169\\
        \multirow{2}{*}{PUSM~\citep{wang-etal-2023-unsup-speech2sign}} & No & Syllable & 28.4 & 26.5 \\
        & Yes & Syllable & 18.5 & 14.9\\
        \midrule
        \sylcipher/ (Ours, forced align) & No & Syllable & 38.1 & 38.9\\
        \sylcipher/ (Ours, Sylber) & No & Syllable & 44.6 & 48.3 \\
        \sylcipher/ (Ours, Sylber+JE2E) & No & Syllable & 41.7 & 45.1 \\
        \multirow{2}{*}{\sylcipher/ (Ours, Sylber+JE2E+PUSM)} & No & Syllable & \textbf{26.9} & \textbf{24.9} \\
        & Yes & Syllable & \textbf{15.3} & \textbf{12.2} \\
        \multirow{1}{*}{\sylcipher/ (Ours, Sylber+JE2E+PUSM)} & No & Char. & 31.9 & 29.9 \\
    \bottomrule
    \end{tabular}}
    \figvsp
    \caption{\textbf{UASR results on AISHELL-3 test set}. Inside the bracket lists the unsupervised boundary used for each model. ``Student'' stands for the student model used during the self-training stage using pseudo-labels from each model. \emph{Init./Final} refers to initials and finals in the Chinese phonetic alphabet. For text data tokens, \emph{Syllable} refers to the pinyin representation of each Chinese character, and \emph{Phone} denotes the individual letters within the pinyin tokens, and \emph{Char.} stands for the raw Chinese character. PER stands for phone error rate.}
    \label{tab:uasr_aishell3}
\end{table}

\begin{table}[t]
    \centering
    \resizebox{0.49\textwidth}{!}{
    \begin{tabular}{lccccc}
    \toprule
        \multirow{1}{*}{\textbf{Model}} & \multirow{1}{*}{\textbf{Token}} & \textbf{de ($\downarrow$)} 
        & \textbf{nl ($\downarrow$)} & 
        \textbf{fr ($\downarrow$)} \\
        \toprule
        \multirow{1}{*}{\wtovu/~\citep{Baevski2021-wav2vec-u}} & Phone & 32.5 & 40.2 & 39.8\\
        \multirow{1}{*}{\wtovu/~\citep{Baevski2021-wav2vec-u}} & Char. & 42.9 & 50.6 & 72.6\\
        PUSM~\citep{wang-etal-2023-unsup-speech2sign} & Syllable & 35.1 & 38.2 & 43.6 \\
        \midrule
        SylCipher (Ours, Sylber) & Syllable & \textbf{38.2} & \textbf{40.9} & \textbf{46.3}\\
        SylCipher (Ours, Sylber+PUSM) & Syllable & \textbf{32.3} & \textbf{35.6} & \textbf{41.7}\\
        \bottomrule
    \end{tabular}}
    \caption{\textbf{Multilingual UASR results on MLS}. \emph{de}, \emph{nl} and \emph{fr} stand for German, Dutch and French respectively. Performance is measured by phone/character error rate (lower is better), depending on the type of token used.}
    \label{tab:mls}
\end{table}

\begin{table}[t]
    \centering
    \figvsp
    \resizebox{0.49\textwidth}{!}{
    \begin{tabular}{llccccc}
    \toprule
        \multirow{1}{*}{\textbf{Model}} &  \multirow{1}{*}{\textbf{Student}} & \multirow{1}{*}{\textbf{Token}} & \multicolumn{1}{c}{\textbf{Cantonese ($\downarrow$)}} & \multirow{1}{*}{\textbf{Taigi ($\downarrow$)}} \\ 
    \midrule
    \midrule
        \multirow{1}{*}{\wtovu/~\citep{Baevski2021-wav2vec-u}}  & No & Phone & 75.7  & 69.0\\ 
        PUSM~\cite{wang-etal-2023-unsup-speech2sign} & No & Syllable & 36.6 & 37.3\\
        \midrule
        \sylcipher/ (Ours, PUSM+MLM) & No & Syllable & 48.8 & 58.4 \\
        \sylcipher/ (Ours, PUSM+MLM+JE2E) & No & Syllable & 37.7 & 38.6\\
        \multirow{2}{*}{\sylcipher/ (Ours, PUSM+MLM+JE2E+PUSM)} & No & Syllable & \textbf{35.0} & \textbf{35.5} \\
        & Yes & Syllable & \textbf{23.5}  & \textbf{26.5}\\
        \multirow{1}{*}{\sylcipher/ (Ours, PUSM+MLM+JE2E+PUSM)} & No & Char. & 46.9 & - \\
    \bottomrule
    \end{tabular}}
    \figvsp
    \caption{\textbf{Multilingual UASR results on two low-resource languages, Cantonese (MDCC) and Taigi (SuiSiann)}.\liming{Multilingual vs. monolingual result} Inside the bracket lists the unsupervised boundary used for each model. ``Student'' stands for the student model used during the self-training stage using pseudo-labels from each model. \emph{Syllable} refers to the JyutPing (or analogous notation for Taigi) representation of each Cantonese/Taigi character, and \emph{Phone} denotes the individual letters within the JyutPing tokens, while \emph{Char.} stands for the raw characters of the languages. Performance are measured by phone error rate (lower is better). \liming{TODO Update results}}
    \label{tab:uasr_low_resource}
    \figvsp
\end{table}

\subsecvsp
\subsection{Results: UASR}\label{sec:uasr_result}
\subsecvsp
We compare \sylcipher/ with word-, character-, and syllable-level UASR systems that differ in tokenization, objective, and architecture: word-level JSTTI~\citep{Ni2025-jstti}; character-level \wtovu/~\citep{Baevski2021-wav2vec-u}, \reborn/~\citep{Tseng2024-reborn}, and a phone-level JSTTI variant; and syllable-level PUSM~\citep{wang-etal-2023-unsup-speech2sign}. We also report \emph{self-training}~\cite{Chen2019-uasr,Baevski2021-wav2vec-u}, where a \wtov/~\cite{Baevski2020-wav2vec2} student is finetuned on pseudo-labels from each UASR system. We use CER for European languages and PER for Sinitic languages, both are tokenization-independent and easily comparable. 

\noindent\textbf{Syllable-level modeling performs best under G2P-free setting.} Table~\ref{tab:uasr_libri} summarizes results on LibriSpeech. Among baselines, PUSM performs best in the matched setting, while \wtovu/ is strongest in the unmatched setting. \reborn/ trained directly on characters degrades substantially, suggesting sensitivity to speech-text misalignment. \sylcipher/ with all three stages (Sylber+JE2E+PUSM) achieves 21.8\% CER (matched) and 35.9\% (unmatched), outperforming all baselines. Relative to \wtovu/ (35.6\%/43.3\%), this is a 40\%/17\% CER reduction. 
Syllable-level models consistently outperform word- and character-level systems, supporting the view that syllables are the most effective alignment unit. Across stages, JE2E brings modest gains, while PUSM yields the largest improvement (44\% and 23\% relative over JE2E in matched and unmatched settings), and
self-training further reduces CER by 20\% relative. Unsupervised Sylber boundaries also perform nearly as well as forced alignment.

\noindent

\noindent\textbf{Syllables are robust to domain shifts.} Table~\ref{tab:uasr_scoco} reports results on SpokenCOCO.  \sylcipher/ remains strongest under domain shift, outperforming PUSM by 32\% and \wtovu/ by 49\% relative CER after all stages; even the early stages already beat \wtovu/ by 22\%. The gap is larger than on LibriSpeech, suggesting that syllable units are more robust to domain shifts for English. Self-training yields a further 39\% relative CER reduction, and the small matched/unmatched gap suggests that the sharper LibriSpeech degradation mainly comes from speech-text domain mismatch.

\noindent\textbf{Syllable-level UASR works under multilingual setting.} Table~\ref{tab:mls} show that multilingual training on related languages is effective: on MLS, \sylcipher/ improves over strong monolingual G2P-free baselines by 25-50\% relative to \wtovu/ and 5-8\% relative to PUSM, and even beats the G2P-based baseline in Dutch by 12\%.

\noindent\textbf{Syllable-level UASR works across (low-resource) languages.} Table~\ref{tab:uasr_aishell3} shows that syllable-level UASR is particularly well-suited to Mandarin: syllable models converge even without boundary refinement, whereas phone- and word-level systems struggle. Relative to PUSM, \sylcipher/ improves PER by more than 5.5\% before self-training and 17\% after; adding tone labels does not hurt, and often helps after the PUSM stage. Similar trends appears in Table~\ref{tab:uasr_low_resource} for Cantonese and Taigi, where \sylcipher/ clearly outperforms the GAN-based character baseline and remain competitive even with raw characters. In these extremely low-resource settings, however, we find it necessary to pretrain with PUSM before MLM (more details in Appendix~\ref{app:implementation}), and the gap to PUSM narrows, likely because limited data makes higher-order speech-text matching harder. Multilingual training also stabilizes the MLM stage; for example, Taigi-only training diverges.

\begin{table}[t]
    \centering
    \begin{subtable}{0.49\textwidth}
    \centering
    \resizebox{0.95\textwidth}{!}{
    \begin{tabular}{l|ccccc}
    \toprule
         & \textbf{F1$^{50}$ ($\uparrow$)} & \textbf{F1$^{20}$ ($\uparrow$)} & \textbf{P. ($\uparrow$)} & \textbf{Re. ($\uparrow$)} & \textbf{R ($\uparrow$)} \\
    \midrule
    \midrule
    Feat-Sim~\citep{Peng2023-syllable}  & 47.3 & 24.7 & 46.6 & 48.0 & 54.4 \\
    SDHuBERT~\citep{cho2023-sdhubert} & 66.1 & 32.2 & 64.9 & 67.4 & 70.7 \\
    SylBoost~\citep{Baade2025-syllablelm} & 73.2 & \underline{44.6} & 72.1 & 74.4 & 76.9\\
    Sylber~\citep{Cho2025-sylber} & \underline{83.4} & 44.1 & \underline{84.8} & \underline{84.1} & \underline{86.4}\\
    \midrule
    SylCipher (Ours, Sylber+JE2E) & \textbf{86.1} & \textbf{50.8} &  \textbf{86.6} & \textbf{86.1} & \textbf{88.1} \\
    \bottomrule
    \end{tabular}}
    \caption{Boundary detection results on LibriSpeech}
    \label{tab:syl_bnd_libri}
    \end{subtable}
    
    \begin{subtable}{0.49\textwidth}
    \centering
    \resizebox{0.95\textwidth}{!}{\begin{tabular}{l|ccccc}
    \toprule
         & \textbf{F1$^{50}$ ($\uparrow$)} & \textbf{F1$^{20}$ ($\uparrow$)} & \textbf{P.} ($\uparrow$) & \textbf{Re.} ($\uparrow$) & \textbf{R ($\uparrow$)} \\
    \midrule
    \midrule
    Feat-Sim~\citep{Peng2023-syllable}  & \underline{60.3} & - & 57.4 & \underline{63.6} & \underline{64.3}\\
    SylBoost~\citep{Baade2025-syllablelm} & 55.6 & \underline{30.2} & 48 &  \textbf{65.2} & 50.8 \\
    Sylber~\citep{Cho2025-sylber} & 53.5 & 28.6 & 52.3 & 54.9 & 59.6 \\
    \midrule
    SylCipher (Ours, Sylber+JE2E) & \textbf{62.3} & \textbf{31.0} & \textbf{61.1} & \underline{63.6} & \textbf{67.4}	\\
    \bottomrule
    \end{tabular}}
    \caption{Boundary detection results on SpokenCOCO}
    \figvsp
    \label{tab:syl_bnd_scoco}
    \end{subtable}
    \caption{\textbf{Syllable boundary detection results on LibriSpeech and SpokenCOCO}. The superscript for F1 scores is tolerance threshold in ms and the tolerance is 50ms for other metrics. P., Re. and R stand for precision, recall and R-value respectively. \sylcipher/ is trained under unmatched settings.}
\end{table}

\begin{figure}
    \begin{subfigure}{0.49\textwidth}
        \centering
        \begin{tabular}{l|cc}
            \toprule
           \textbf{Domain} & \textbf{Ratio} & \textbf{CER} \\
           \midrule
            None & 0.0 & 33.0 \\
            News & 0.5 & 39.4 \\
            News & 0.75 & 48.3 \\
            Wikipedia & 0.5 & 42.8 \\
            Wikipedia & 0.75 & 57.1 \\
            Image Caption & 0.25 & 37.0 \\
            Image Caption & 0.5 & 75.6 \\
            \bottomrule
        \end{tabular}
        \caption{\textbf{Effect of speech-text domain mismatch.} A PUSM with XEUS backbone is trained on mixed LibriSpeech text with text from three out-of-domain (OOD) sources: news (from NewsCrawl), Wikipedia and image captions (from MS COCO), using different proportions of OOD text.}
        \label{tab:eff_domain_mismatch}
    \end{subfigure}
    \begin{subfigure}{0.49\textwidth}
        \centering
        \resizebox{0.95\textwidth}{!}{\begin{tabular}{l|cc}
        \toprule
            \textbf{Backbone} & \textbf{LibriSpeech ($\downarrow$)} & \textbf{MLS (german) ($\downarrow$)} \\
            \midrule
            HuBERT & 35.5 & 45.5 \\
            XEUS & \textbf{33.4} & \textbf{35.1} \\
            \bottomrule
        \end{tabular}
        }
        \caption{\textbf{Effect of speech encoder backbone choices.} The model used is a syllable-level PUSM and the evaluation metric used is CER.}
        \label{tab:eff_of_backbone}
    \end{subfigure}
    \caption{\textbf{Ablation studies} on the effect of speech-text domain mismatch and speech encoder backbone choices for syllable-level UASR.}
    \label{fig:ablation_domain_and_backbone}
\end{figure}

\begin{figure*}
    \centering
    \begin{subfigure}{0.32\textwidth}
        \includegraphics[width=0.99\textwidth]{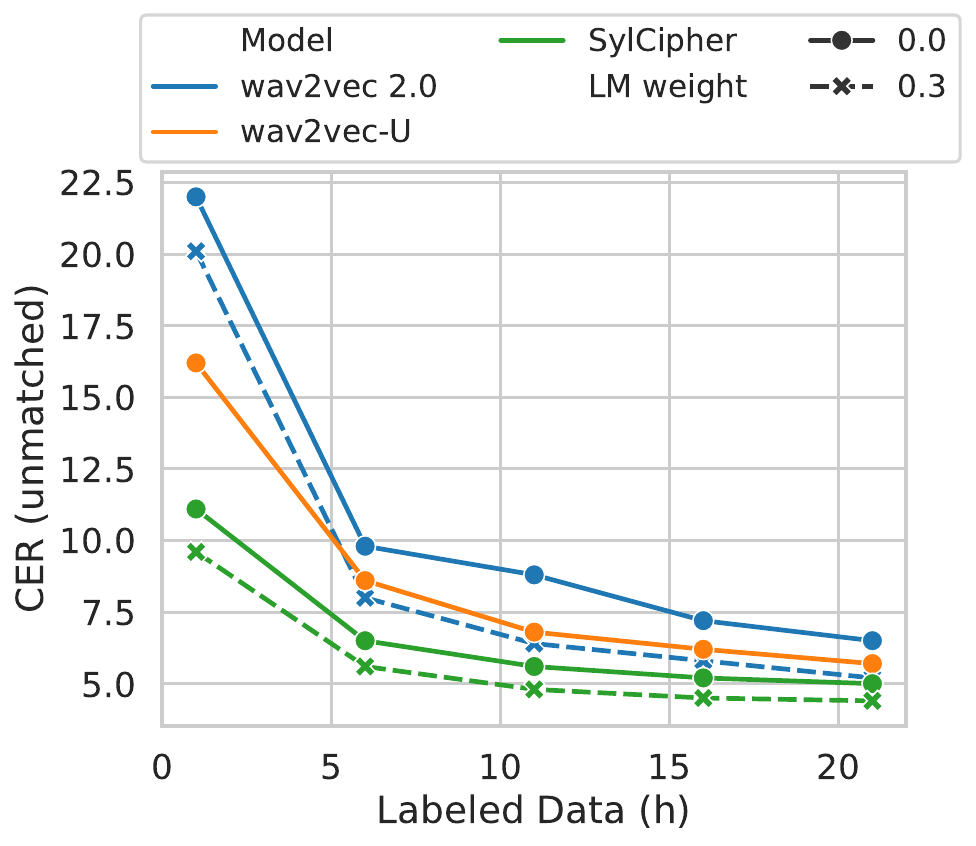}
        \caption{LibriSpeech (English)}
    \end{subfigure}
    \begin{subfigure}{0.32\textwidth}
        \includegraphics[width=0.99\textwidth]{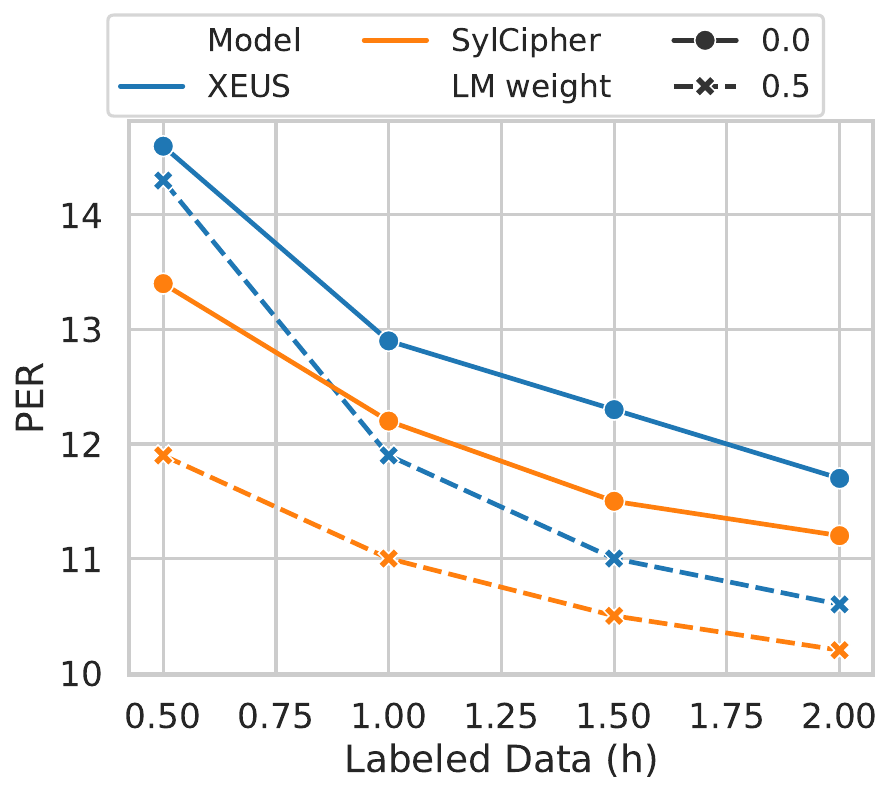}
        \caption{MDCC (Cantonese)}
    \end{subfigure}
    \begin{subfigure}{0.32\textwidth}
        \includegraphics[width=0.99\textwidth]{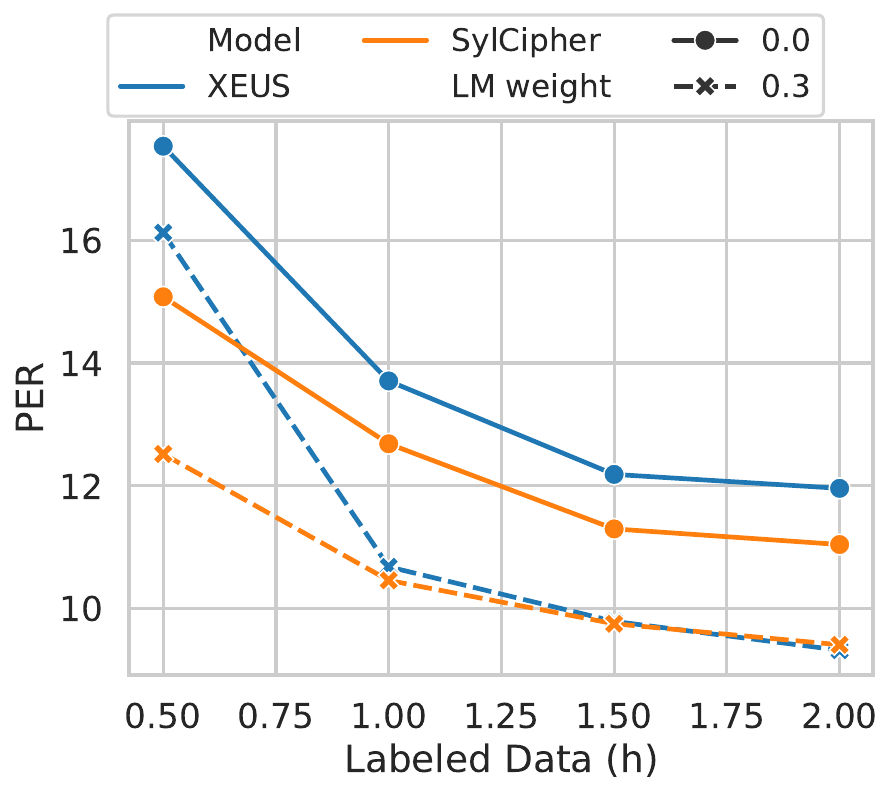}
        \caption{SuiSiann (Taigi)}
    \end{subfigure}
    \caption{\textbf{Semi-supervised ASR results on various spoken languages}. A small subset of labeled data is randomly sampled from the training set to simulate the low-resource setting. ``LM weight'' stands for the fusion weight of a 4-gram language model during beam search. For each dataset, all methods share the same model architecture and hyperparameter settings, except our approach uses an additional unpaired speech-text self-training stage before the supervised finetuning. Lower is better for all metrics.}
    \label{fig:semi_sup}
\end{figure*}
\begin{figure}[t]
\centering
    \begin{subfigure}{0.23\textwidth}
        \centering
        \includegraphics[width=0.99\textwidth]{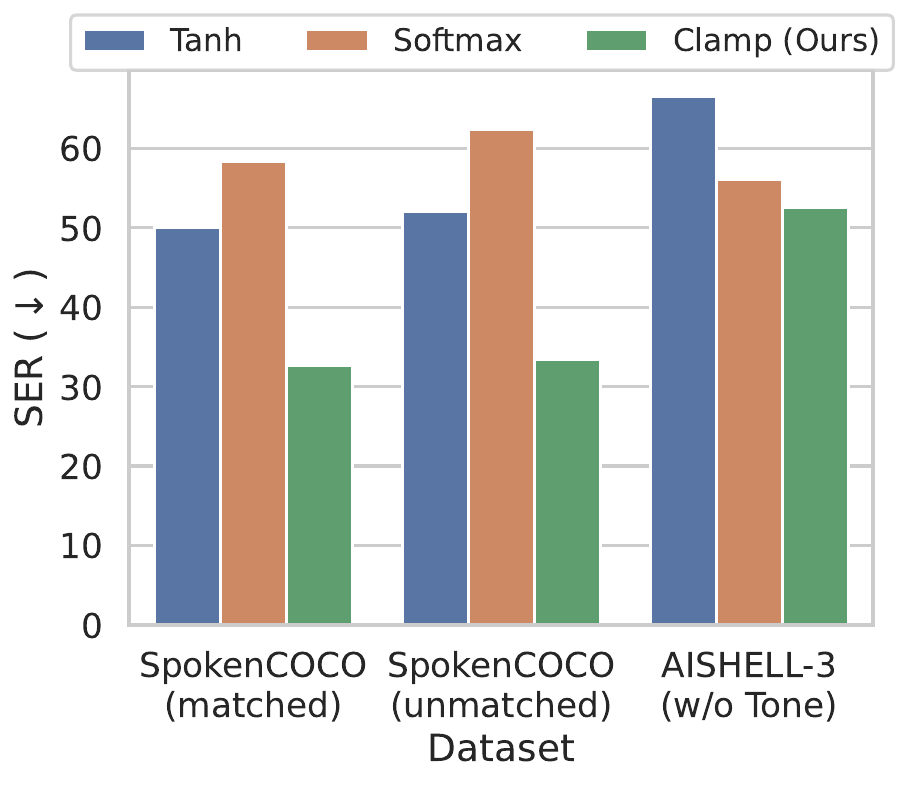}
        \figvsp
        \caption{SER vs. pooler type across datasets}
        \label{fig:ser_vs_pooler}
    \end{subfigure}
    \begin{subfigure}{0.23\textwidth}
        \centering
        \includegraphics[width=0.99\textwidth]{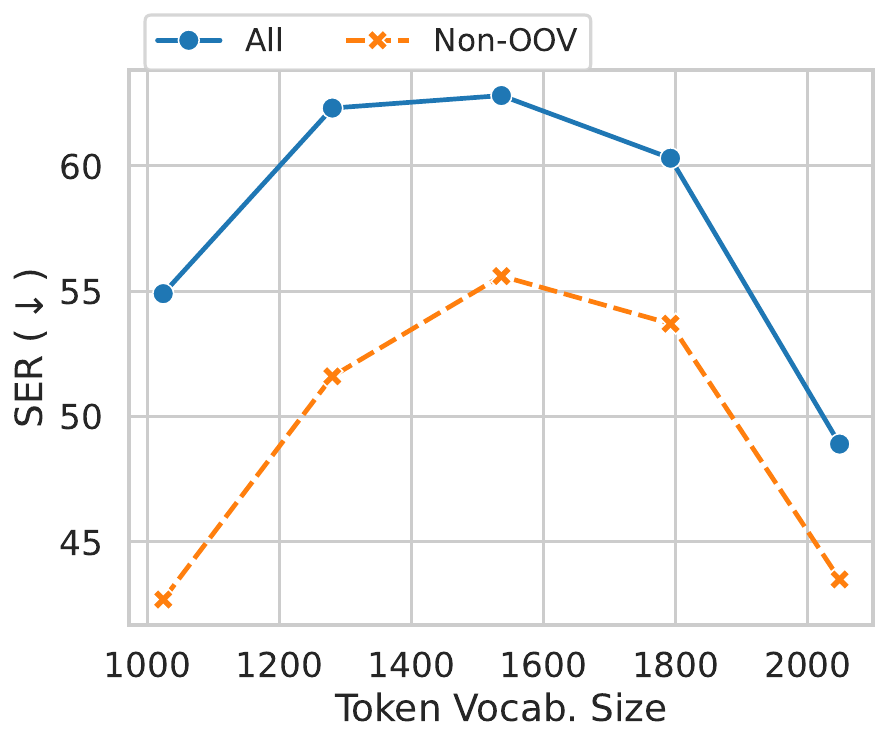}
        \figvsp
        \caption{SER vs. token vocabulary size}
        \label{fig:ser_vs_token_vocab}
    \end{subfigure}

    \begin{subfigure}{0.49\textwidth}
    \centering
    \resizebox{0.95\textwidth}{!}{
    \begin{tabular}{lcccc}
         \toprule
         \textbf{Model} & \multirow{1}{*}{\textbf{Tokenizer}} & \multirow{1}{*}{\textbf{Language}} & \textbf{SER ($\downarrow$)} & \textbf{CER ($\downarrow$)}\\
         \midrule
         \midrule
         \sylcipher/ (Sylber) & \pyphen/+ & English & 48.9 & -\\
         \sylcipher/ (Sylber+JE2E) & \pyphen/+ & English & 44.9 & -\\
         \sylcipher/ (Sylber) & Syllabify & English & 51.3 & -\\
         PUSM & \pyphen/+ & German & 46.3 & 35.1 \\
         PUSM & \pyphen/+ & Dutch & 48.7 & 38.2 \\
         PUSM & \pyphen/+ & French & 52.7 & 43.6 \\
         \midrule
         \sylcipher/ (Sylber) & BPE+ & English & 52.5 & - \\
         \sylcipher/ (Sylber+JE2E) & BPE+ & English & 49.9 & - \\
         PUSM & BPE+ & German & 54.3 & 39.7	 \\
         PUSM & BPE+ & Dutch & 54.4 & 39.8	\\
         PUSM & BPE+ & French & 64.3 & 48.3 \\
         \bottomrule
    \end{tabular}}
    \caption{Effect of syllabifier type}
    \label{tab:ablation}    
    \end{subfigure}
    \caption{\textbf{Ablation studies} on the effect of syllabifier type, soft-pooler design and token vocabulary size on \sylcipher/ UASR performance.}
    \label{fig:ablation}
    \figvsp
\end{figure}

\subsecvsp
\subsection{Results: Unsupervised syllable boundary detection}\label{sec:bd_result}
\subsecvsp
To isolate the effect of JE2E, we compare its boundary predictions against its teacher Sylber~\citep{Cho2025-sylber}  and other unsupervised approaches including Feat-Sim~\citep{Peng2023-syllable}, SDHuBERT~\citep{cho2023-sdhubert}, SylBoost~\citep{Baade2025-syllablelm}. Sylber is the strongest baseline on most LibriSpeech metrics, while SylBoost is strongest on LibriSpeech F1$^{20}$ and three of five SpokenCOCO metrics. However, SylBoost often over-segments, which hurts cross-modal alignment and destabilizes \sylcipher/, whereas Sylber predicts syllable counts closer to ground truth and is therefore a better initializer. JE2E improves over Sylber by +14\% relative F1$^{20}$ and +3\% relative F1$^{50}$ on LibriSpeech, and by +15\% F1 and +11\% R-value on SpokenCOCO. It also outperforms the best SpokenCOCO baseline, Feat-Sim, by +3\% relative F1 score and +4.6\% R-value. This suggests that unpaired text provides useful boundary supervision. Speech-text alignments are shown in Appendix~\ref{app:spec_libri}.

\subsection{Results: Semi-supervised ASR}
We also evaluate a semi-supervised setting with limited paired data by finetuning the self-training student models on labeled subsets and comparing them with the same students before self-training. At inference time, we optionally use a 4-gram LM and tune dataset-specific decoding hyperparameters; a beam size of 5 is sufficient throughout. Figure~\ref{fig:semi_sup} shows that, without an LM, \sylcipher/ improves over the baselines by 20--100\% on English with up to 20 hours of labeled data and by roughly 7--10\% on Cantonese and Taigi with up to 2 hours. The same trend holds with an LM, although the gap closes faster as more labeled data becomes available. These results highlight the value of decipherment-based pretraining in low-resource supervised ASR.

\subsecvsp
\subsection{Ablation studies}\label{sec:ablation}
\subsecvsp
We ablate the effects of speech-text domain mismatch and speech encoder choice on syllable-level UASR in Figure~\ref{fig:ablation_domain_and_backbone} and the effects of syllabifier, soft-pooler, and vocabulary size in Figure~\ref{fig:ablation}.

\noindent\textbf{Level of speech-text domain mismatch.} As shown in Figure~\ref{tab:eff_domain_mismatch}, replacing part of the training corpus with out-of-domain text consistently degrades performance, with the deterioration becoming more pronounced as the proportion of out-of-domain text increases. News text raises the CER to 39.4 and 48.3 at ratios of 0.5 and 0.75, respectively, while Wikipedia text produces even larger increases to $42.8\rightarrow 57.1$. Image captions cause the most severe degradation: although a ratio of 0.5 and 0.75, respectively, while Wikipedia text produces even larger increases to $42.8\rightarrow 57.1$. Image captions cause the most severe degradation: although a ratio of 0.25 results in a moderate CER of 37.0, increasing the ratio to 0.5 raises the CER dramatically to 75.6. These results show that syllable-level UASR is sensitive not only to the amount of mismatched text but also to the nature and severity of the domain mismatch.

\noindent\textbf{Speech backbone choice.} As shown in Figure~\ref{tab:eff_of_backbone}, replacing HuBERT with XEUS reduces the CER from 35.5 to 33.4 on LibriSpeech and from 45.5 to 35.1 on German MLS. The considerably larger improvement on German MLS suggests that XEUS provides more robust and transferable speech representations, particularly in the cross-lingual setting.

\noindent\textbf{Syllabifier choice.} We compare Pyphen+ with the phoneme-based Syllabify\footnote{https://github.com/kylebgorman/syllabify} and a lightweight approach BPE-based syllabifier (BPE+) derived from the first stage of SuperBPE~\cite{liu-etal-2025-superbpe} with a vocabulary of 17k, roughly matching the number of syllable types in LibriSpeech. More details are in Appendix~\ref{app:code_bpe+}. All variants are trained only in the \emph{fixed-boundary stage}, and report syllable error rate (SER) to focus on syllable-level behavior. Pyphen+ gives the best SER, outperforming even Syllabify, while BPE+ remains competitive, showing that \sylcipher/ is robust even with simple, resource-light tokenizers.

\noindent\textbf{Soft-pooler design.} Figure~\ref{fig:ser_vs_pooler} compares our clamp-based soft-pooler with $\tanh$- and softmax-based alternatives~\citep{Bhati2022-scpc,wang2024unsupervised}. Using $\epsilon=0.5$ for Clamp and $0.1$ for the others, Clamp consistently improves SER by 6-36\% on SpokenCOCO and \aishell/, indicating more stable training.

\noindent\textbf{Vocabulary size.} Figure~\ref{fig:ser_vs_token_vocab} shows that \sylcipher/ is robust across a wide range of non-\texttt{<OOV>} vocabulary sizes, with the best overall SER and best non-\texttt{<OOV>} SER at 2048 and 1024 tokens.

%% file: sections/conclusion.tex
In this work, we introduced \sylcipher/, a UASR system that avoids phoneme-level resources such as G2P by recognizing speech at the \emph{syllable level}. Under the G2P-free setting, \sylcipher/ outperforms the best existing systems by 17-40\% relative CER on LibriSpeech and shows generalizability across other domains on SpokenCOCO, narrowing the gap with G2P-based systems. It also demonstrates cross-lingual robustness, including Mandarin and two low-resource languages where phoneme-based methods fail to converge, \sylcipher/ achieves a PER of 13\%. These results suggest that syllable-level modeling is a viable alternative to phoneme-level UASR and can push the horizon of more accessible and inclusive spoken language technology.


%% file: sections/appendix.tex
\section{Theoretical guarantee of \sylcipher/}\label{app:proof_of_main}
\liming{fix appendix}We state our main theorem below.
\begin{theorem}\label{thm:main}
    Let $(f_{\tilde{X}}^*, g_{\tilde{X}}^*, f_Y^*, g_Y^*)$ be a minimizer of \eqref{eq:regularized_distribution_matching} under the constraint that $H(f_{Z}(Z))\leq H(Y)$, and suppose the following assumptions hold:
    \begin{enumerate}
        \item The speech feature sequence $X$ is a sequence of one-hot vectors with $|\gX|=|\tilde{\gX}|=|\gY|$; 
        \item The true syllable boundaries are used, i.e., $b_t(X)=\mathbbm{1}[y^*(X_t)\neq y^*(X_{t+1})]$;
        \item The tokenizer $c$ is an optimal \km/ quantizer with Euclidean distance metric and cluster size $|\tilde{\gX}|$;
        \item The true ASR $y^*$ is decomposable, i.e., $y^*(\tX)=[y^*(\tX_1),\cdots,y^*(\tX_L)]$;
        \item $f_{\tX}$ and $f_Y$ are decomposable;
        \item Assumption 1-2 in \citep{wang-etal-2023-unsupasr-theory} hold for $(\tilde{X}, Y)$.
    \end{enumerate}
    Then $f_{\tX}^*$ and $f_Y^*$ are invertible and  $q_{Y|X}^*(y|x):=\mathbbm{1}[f_Y^{*-1}\circ f_{\tilde{X}}^*\circ c \circ m(x)=y]$ satisfies 
    \begin{align*}
    &\KL(p_Y||\E_X[q_{Y|X}^*])=0,\\
    &y^*(x)=\argmax_{y\in \gY^L} q^*_{Y|X}(y|x),\quad\forall x\in\gX^T.
    \end{align*}
\end{theorem}

The proof relies on the following lemma proven in Appendix~\ref{app:proof_of_lemma_invertible}.
\begin{lemma}\label{lemma:invertible}
Suppose Assumption 5 of Theorem~\ref{thm:main} and  Assumption 1 and 2 of \citep{wang-etal-2023-unsupasr-theory} hold for $(\tX,Y)$, and functions $g_X:\tilde{\gX}^L\mapsto [0,1],g_Y:\tilde{\gY}^L\mapsto[0,1]$ of the form in \eqref{eq:postnet} be such that $p_X=g_X\circ f_X$ and $p_Y=g_Y\circ f_Y$. Then $f_X$ and $f_Y$ are both invertible.
\end{lemma}

Now we are ready to prove the main theorem. 
\begin{proof}
The proof consists of two main parts: (i) First, we prove that at least one  minimizer $(f_{\tilde{X},1}, g_{\tilde{X},1}, f_{Y,1}, g_{Y,1})$ can achieve a minimum of 0 for \eqref{eq:regularized_distribution_matching}; (ii) then we establish that for any minimizers $(f_{\tilde{X}}^*, g_{\tilde{X}}^*, f_Y^*, g_Y^*)$, the marginal distribution of the text posterior $q_{Y|X}^*$ matches the true text distribution $p_Y$. As a result, by Assumption 4, Theorem 1 of \citep{wang-etal-2023-unsupasr-theory} guarantees that $\argmax_{y\in\gY^L}q_{Y|X}^*(y|X)$ achieves zero-error UASR. 

To prove (i), by the definition of $y^*$ and Assumption 1, $y^*$ is a invertible mapping between $\gX$ and $\gY$. Then let $t_i=\min\{\tau:\sum_{t=1}^{\tau}b_t(X)\geq i-1\}$ be the \emph{starting time} of each syllable and apply Assumption 2, we have
\begin{align*}
    X_s = X_t,\,\forall t_i\leq s<t< t_{i+1},\,1\leq i\leq L.
\end{align*}
As a result, the output of the soft-pooler is simply $m_i(X)=X_{t_i}.$ Further, by Assumption 3, $c$ is optimal, and we claim that this is achievable if and only if for any $(x,x')\in\gX^2,$ 
\begin{align*}
    c(x)=c(x')\Longleftrightarrow y^*(x)=y^*(x').
\end{align*}
Otherwise, suppose for some $(x,x')\in\gX^2$ such that $c(x)=c(x')=:c_0$ but $y^*(x)\neq y^*(x'),$ then for any centroid $\mu_{c_0}$ of cluster $c_0$ and triangle inequality,  
\begin{align*}
    &\|x-\mu_{c(x)}\|_2+\|x'-\mu_{c(x')}\|_2\\
    =&\|x-\mu_{c_0}\|_2+\|x'-\mu_{c_0}\|_2\geq \|x-x'\|_2 >0,
\end{align*}
and therefore the \km/ objective:
\begin{align*}
    \Ls_{\mathrm{km}}(\mu_1,\cdots,\mu_{|\gX|}):=\E_{X\sim p_X}\|X-\mu_{c(X)}\|^2_2>0.
\end{align*}

However, by simply setting $c(x)=y^*(x)$ for all $x\in\gX$, we have $\mu_{c(x)}=x$ for any $x\in \gX$ and $\Ls_{\mathrm{km}}=0$ due to the invertibility of $y^*.$ This contradicts with the optimality of $c$ and thus proves the ``$\Longrightarrow$'' direction. Further, this implies that the cluster size is at least $|\gY|$ and since we set the cluster size to $|\gY|$, $c\circ m$ is equivalent to $y^*$ up to a permutation. This then proves the other direction. Let $c\circ m=\pi\circ y^*$ where $\pi$ is a permutation, then the syllable-level features $\tilde{X}=\pi\circ y^*(X)$ and thus by the permutation-invariance of discrete entropy, 
\begin{align*}
    H(\tilde{X})=H(y^*(X))=H(Y).
\end{align*}
Therefore, set $f_{\tilde{X},1}$ to be a permutation on $\tilde{\gX}$, $f_{Y,1}=f_{\tilde{X},1}\circ\pi$, $g_{\tilde{X},1}=p_{\tilde{X}}\circ f_{X,1}^{-1}$ and $g_{Y,1}=p_Y\circ f_{Y,1}^{-1},$ we have
\begin{align*}
    &\KL(p_{\tilde{X}}||g_{\tilde{X},1}\circ f_{\tilde{X},1})+\KL(p_Y||g_{Y,1}\circ f_{Y,1})\\
    =&\KL(p_{\tilde{X}}||p_{\tilde{X}})+\KL(p_Y||p_Y)=0,\\
\end{align*}
and 
\begin{align*}
    &f_{Y,1}(y(X))\\
    =&f_{\tilde{X},1}\circ \pi\circ y(X)=f_{\tilde{X},1}\circ c\circ m(X)\\
    =&f_{\tilde{X},1}(\tilde{X}).
\end{align*}
The latter implies that the \pdf/ of $f_{Z,1}(Z)$ satisfies
\begin{align}
p_{f_{Z,1}(Z)}&=\frac{1}{2}p_{f_{\tilde{X},1}(\tilde{X})}+\frac{1}{2}p_{f_{Y,1}(Y)}\nonumber\\
&=\frac{1}{2}p_{f_{\tilde{X},1}(\tilde{X})}+\frac{1}{2}p_{f_Y(y(X)),1}=p_{f_{\tilde{X},1}(\tilde{X})},\tag{*}    
\end{align}
and thus 
$$H(f_{Z,1}(Z))=H(f_{Y,1}(Y))=H(f_{X,1}(\tX))\leq H(Y).$$
Therefore, by the nonnegativity of $\KL$, $(f_{\tilde{X},1}, g_{\tilde{X},1}, f_{Y,1}, g_{Y,1})$ is an optimal solution of \eqref{eq:regularized_distribution_matching} with a minimum of $0$, which proves $(i)$. 

To prove (ii), we first use (i) to conclude that
\begin{align*}
    p_{\tX}=g_{\tX}^*\circ f_{\tX}^*,\,p_Y=g_Y^*\circ f_Y^*.
\end{align*}
Then it amounts to prove that any minimizer $(f_{\tilde{X}}^*,g_{\tilde{X}}^*,f_Y^*,g_Y^*)$ of \eqref{eq:regularized_distribution_matching} satisfies
\begin{align}
    p_{f_Z^*(Z)}=p_{f_X^*(X)}=p_{f_Y^*(Y)}.\tag{*}
\end{align}
Since if this is the case, by Lemma~\ref{lemma:invertible}, $f_{Y}^*$ is invertible and thus for any $y\in\gY^L$,
\begin{align*}
    &p_{f_Y^{*-1}\circ f_X^*(X)}(y)\\
    \overset{(a)}{=}&p_{f_X^*(X)}(f_Y^*(y))
    \overset{(b)}{=}p_{f_Y^*(Y)}(f_Y^*(y))\\
    \overset{(c)}{=}&p_Y(y)\Longrightarrow \KL(p_Y||\E_Xq^*_{Y|X})=0,
\end{align*}
where $(b)$ uses (*) and $(a)(c)$ uses Lemma~\ref{lemma:invertible}.

To prove (*), we use the concavity of the discrete entropy $h(p):=-\sum_xp(x)\log p(x)$ and Lemma~\ref{lemma:invertible},
\begin{align*}
    H(f_Z^*(Z))&=h\brac{\frac{1}{2}p_{f_X^*(X)}+\frac{1}{2}p_{f_Y^*(Y)}}\\
    &\overset{(d)}{\geq} \frac{1}{2}h(p_{f_X^*(X)})+\frac{1}{2}h(p_{f_Y^*(Y)})\\
    &\overset{(e)}{=} \frac{1}{2}(H(X)+H(Y))=H(Y),
\end{align*}
with equality if and only if $p_{f_Y^*(Y)}=p_{f_X^*(X)},$ where $(d)$ uses the concavity of the entropy function and $(e)$ uses Lemma~\ref{lemma:invertible}. This concludes the proof of $(ii).$
\end{proof}

\section{Proof of Lemma~\ref{lemma:invertible}}\label{app:proof_of_lemma_invertible}
\begin{proof}
We prove the lemma by contradiction and focus on proving the invertibility of $f_X$ since the proof is analogous for $f_Y$. Suppose otherwise $f_X$ is not invertible, then by Assumption 5, there exists $(x',x'')\in\tilde{\gX}^2$ such that $f_X(x')=f_X(x'')$ but $x'\neq x''$, then by the definition of $g_X$,
\begin{multline}
    p_{X_i}(x')
    = \sum_{x_{-i}\in\tilde{\gX}^{L-1}}
    p_X(x_{1:i-1},x',x_{i+1:L}) \\
    = \sum_{x_{-i}\in\tilde{\gX}^{L-1}}
    \prod_{j<i}
    g_X\!\left(
        x_j \mid f_X(x_{1:j-1})\right)\cdot\\
    g_X\!\left(
        x' \mid f_X(x_{1:i-1})
    \right)\cdot\\
    \prod_{k>i}
    g_X\!\left(
        x_k \mid
        f_X(x_{1:i-1}), f_X(x'), f_X(x_{i+1:k-1})
    \right) \\
    = \sum_{x_{-i}\in\tilde{\gX}^{L-1}}
    \prod_{j<i}
    g_X\!\left(
        x_j \mid f_X(x_{1:j-1})
    \right)\cdot\\
    g_X\!\left(
        x'' \mid f_X(x_{1:i-1})
    \right)\cdot\\
    \prod_{k>i}
    g_X\!\left(
        x_k \mid
        f_X(x_{1:i-1}), f_X(x''), f_X(x_{i+1:k-1})
    \right) \\
    = \sum_{x_{-i}\in\tilde{\gX}^{L-1}}
    p_X(x_{1:i-1},x'',x_{i+1:L})
    = p_{X_i}(x'').
\end{multline}
Therefore, the positional unigram matrix (PUM)
\begin{align*}
    P^X:=\begin{bmatrix}
        p_{X_1}^\top \\
        \vdots \\
        p_{X_L}^\top
    \end{bmatrix}
\end{align*}
is column-rank-deficient. However, Theorem 1 of \citep{wang-etal-2023-unsupasr-theory} asserts that if Assumption 1 and 2 of \citep{wang-etal-2023-unsupasr-theory} holds, $P^X$ has full column-rank, which is a contradiction. Therefore, $f_X(x')\neq f_X(x'')$ if $x'\neq x''$ and $f_X$ is invertible.
\end{proof}

\begin{figure*}
    \centering
    \begin{subfigure}{0.32\textwidth}
        \includegraphics[width=0.95\textwidth]{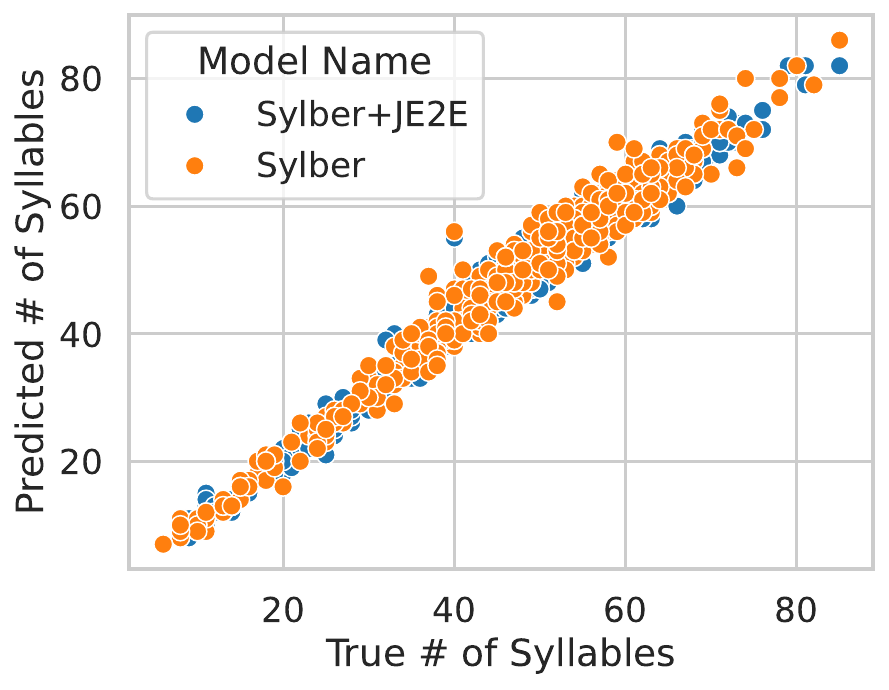}
        \caption{Number of predicted vs. gold syllables in single utterances}
        \label{fig:syl_bound_stat_a}
    \end{subfigure}
    \begin{subfigure}{0.32\textwidth}
        \includegraphics[width=0.95\textwidth]{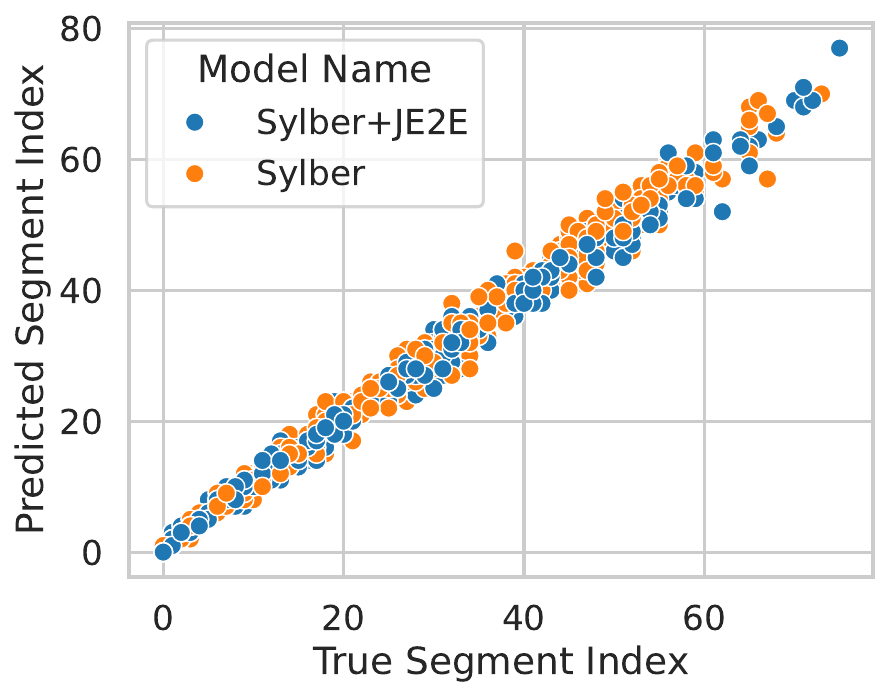}
        \caption{Closest predicted vs. gold syllable segment index}
        \label{fig:syl_bound_stat_b}
    \end{subfigure}
    \begin{subfigure}{0.32\textwidth}
        \includegraphics[width=0.95\textwidth]{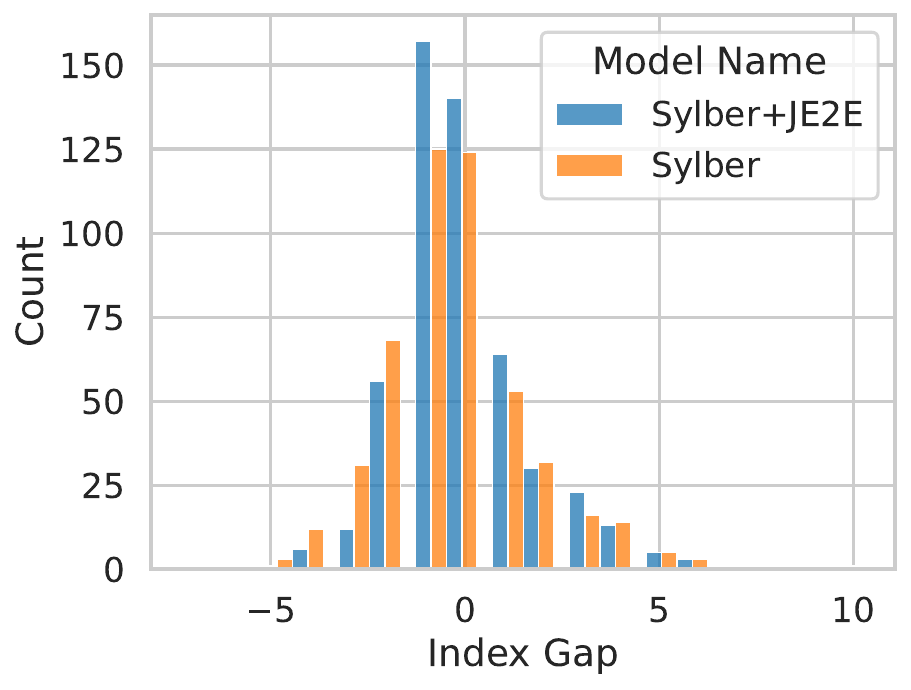}
        \caption{Distribution of the  gold-predicted segment index gap}
        \label{fig:syl_bound_stat_c}
    \end{subfigure}
    \caption{\textbf{Reliability analysis of Syllable-level boundaries}. The syllable boundaries are extracted from the LibriSpeech test sets under the matched setup. The true segment index of a timestamp is the position index of the true segment that it belongs to, and only timestamps corresponding to a predicted syllable boundary is considered in the plots. The index gap is the difference between the true and predicted segment index.}
    \label{fig:syl_bound_stat}
\end{figure*}

\section{Theoretical analysis of \sylcipher/ under clustering and segmentation noise}\label{app:assumption_discussion}
In practice, the speech cluster sequence $X$ can be noisy such that clusters within each syllabic segment $\tilde{X}_i$ may not all correspond to the true syllable $Y_i$. These type of cluster noise can be well-modeled by an $N$-ary symmetric channel~\cite{cover2006-it} defined below.
\begin{definition}{($N$-ary symmetric channel (NSC))}
    An $N$-ary symmetric channel $\mathrm{NSC}(N, p)$ with input $X$ and output $Y$ satisfy
    \begin{align*}
        \Pr[Y=y]=\begin{cases}
            1-p,\,\text{if }y=X,\\
            \frac{p}{N},\,\text{otherwise.}
        \end{cases}
    \end{align*}
    Further, define the \emph{transition matrix} of the channel as
    \begin{align*}
        \Pi_{N,p} := \begin{bmatrix}
            1-p & p/N &\cdots & p/N\\
            p/N & 1-p & \cdots & p/N\\
            \vdots & \vdots & \ddots & \vdots \\
            p/N & p/N & \cdots & 1-p 
        \end{bmatrix}. 
    \end{align*}
\end{definition}

To simplify our analysis, we choose the soft pooler $m(\cdot)$ to output a random feature frame in each predicted segment $i$, $X_{t_i'}$,  rather than a weighted average of all the frame. Then we pass the feature frames independently into an $\mathrm{NSC}(|\gX|, \epsilon_c)$ channel with error probability $\epsilon_c>0$ to create $\tX_i$. Therefore by definition,
\begin{align*}
    \Pr[\tX_i=x]=\begin{cases}
        1-\epsilon_c,\,\text{if }x=X_{t_i'},\\
        \frac{\epsilon_c}{|\gX|},\,\text{otherwise}.
    \end{cases}
\end{align*}
Assuming no segmentation noise, the new speech distribution $\tilde{p}_{\tilde{X}}$ satisfies
\begin{align*}
    \tilde{p}_{\tilde{X}}(x)=\sum_{x'\in \gX^L}p_{\tilde{X}}(x')\Pi_{|\gX|,\epsilon_c}^{\otimes L}(x',x),
\end{align*}
where $\otimes$ denotes the Kronecker product. We can bound the KL divergence between the $\tilde{P}_{\tilde{X}}$ and $P_{\tilde{X}}$ using Jensen's inequality as:
\begin{align*}
    \KL(p_{\tX}||\tilde{p}_{\tX}) &\leq \E_{p_{\tX}}\KL(I_{|\gX|^{L}}||\Pi_{|\gX|,\epsilon_c}^{\otimes L})\\
    &= L\KL(I_{|\gX|}||\Pi_{|\gX|,\epsilon_c})\\
    &=L|\gX|\log \frac{1}{1-\epsilon_c}=O(L|\gX|\epsilon_c).
\end{align*}
Thus, by data processing inequality,
\begin{align*}
    \KL(p_{Y}||\E_{\tilde{p}_{\tX}}q_{Y|X}^*) \leq \KL(p_{\tX}||\tilde{p}_{\tX}) =O(L|\gX|\epsilon_c).
\end{align*}

Further, for small enough  $\epsilon_c > 0$, We can prove a variant of Lemma~\ref{lemma:invertible} using the noisy PUM in place of the clean PUM, the former of which is invertible if and only if the latter is. Therefore, the model is robust against NSC noise and both conclusions from Theorem~\ref{thm:main} hold approximately.

To further verify that the assumptions in Theorem~\ref{thm:main} approximately hold in practice and to develop a more general theory of UASR under segment noise, we visualize the key statistics of the syllable boundaries predicted by Sylber before and after JE2E refinement used in \sylcipher/, as shown in Figure~\ref{fig:syl_bound_stat}.  Figure~\ref{fig:syl_bound_stat_a} indicates that the number of predicted and true syllables in each utterance is approximately the same, even for long utterances with more than 80 syllables. Further, Figure~\ref{fig:syl_bound_stat_b} and Figure~\ref{fig:syl_bound_stat_c} show that the positions of the true and predicted syllables tend to match or within 5 syllables from each other, regardless of the absolute positions of the syllables, suggesting that the insertion/deletion probability of the syllables is low. Motivated by these observations, we propose the following \emph{deletion-replication channel}~\cite{Drinea2007} model for the syllable segmentation noise.
\begin{definition}{(deletion-replication channel (DRC))}\label{def:anchored_noise}
    A deletion-replication channel $\mathrm{DRC}(P_T)$ with input sequence $X=(X_1,\cdots,X_N)$ and output sequence $Y=(Y_1,\cdots,Y_{T})$ satisfies
    \begin{itemize}
        \item $T=\sum_{i=1}^N T_i$, where $T_i \in \{0,\cdots,T_{\max}\}$'s are sampled i.i.d from some distribution $P_T$;
        \item For any $i\in\{1,\cdots,N\}$ and $\sum_{j<i}T_j \leq t \leq \sum_{j\leq i}T_j,$ $$Y_t\equiv X_i.$$
    \end{itemize}
    Further, denote $\epsilon_d:=P_T(0)$ as the \emph{deletion probability} and $\epsilon_r:=\sum_{t>1} P_T(t)$ as the \emph{replication probability}.
\end{definition}
Intuitively, DRC models how the boundary detector i) fails to detect a boundary, leading to deletion and ii) detects multiple boundaries per syllable, leading to replication. Further, from Figure~\ref{fig:syl_bound_stat_c}, we can see that $\epsilon_d$ and $\epsilon_r$ of the syllable segment sequence are both very low. To analyze the effect of DRC on distribution matching, denote the distribution of any $N$-gram subsequence of the predicted segment sequence $\tilde{X}$ as $\tilde{P}_{\tilde{X}_{-N}\cdots \tilde{X}_0}$ then we have 
\begin{multline*}
    \tilde{p}_{\tilde{X}_{-N}\cdots \tilde{X}_0} = (1-\epsilon_r-\epsilon_d)^N p_{\tilde{X}_{-N}\cdots \tilde{X}_0} +\\
    (1-(1-\epsilon_r-\epsilon_d)^N)q,
\end{multline*}
where $p_{\tilde{X}_{-N}\cdots \tilde{X}_0}$ is the true $N$-gram distribution under true syllable boundaries, and $q$ is some residual distribution created when there is at least one insertion or replication in the $N$-gram sequence. Then, as long as $(1-\epsilon_r-\epsilon_d)^N \approx 1$, we have $\KL(P_{\tilde{X}_{-N}\cdots \tilde{X}_0}||\tilde{P}_{\tilde{X}_{-N}\cdots \tilde{X}_0})$ to be low. This also indicates that lower-order N-grams are easier to match than higher-order ones even with infinite data under DRC segmentation noise, as the weight for the true $N$-gram distribution, $(1-\epsilon_r-\epsilon_d)^N$,  decays exponentially as $N$ increases. Nonetheless, using similar arguments, we can show that given enough samples and with sufficiently low $\epsilon_d$ and $\epsilon_r$, one can achieve approximate distribution matching for lower-order $N$-grams and almost zero UASR error rate.

\begin{figure}
    \centering
    \includegraphics[width=0.5\textwidth]{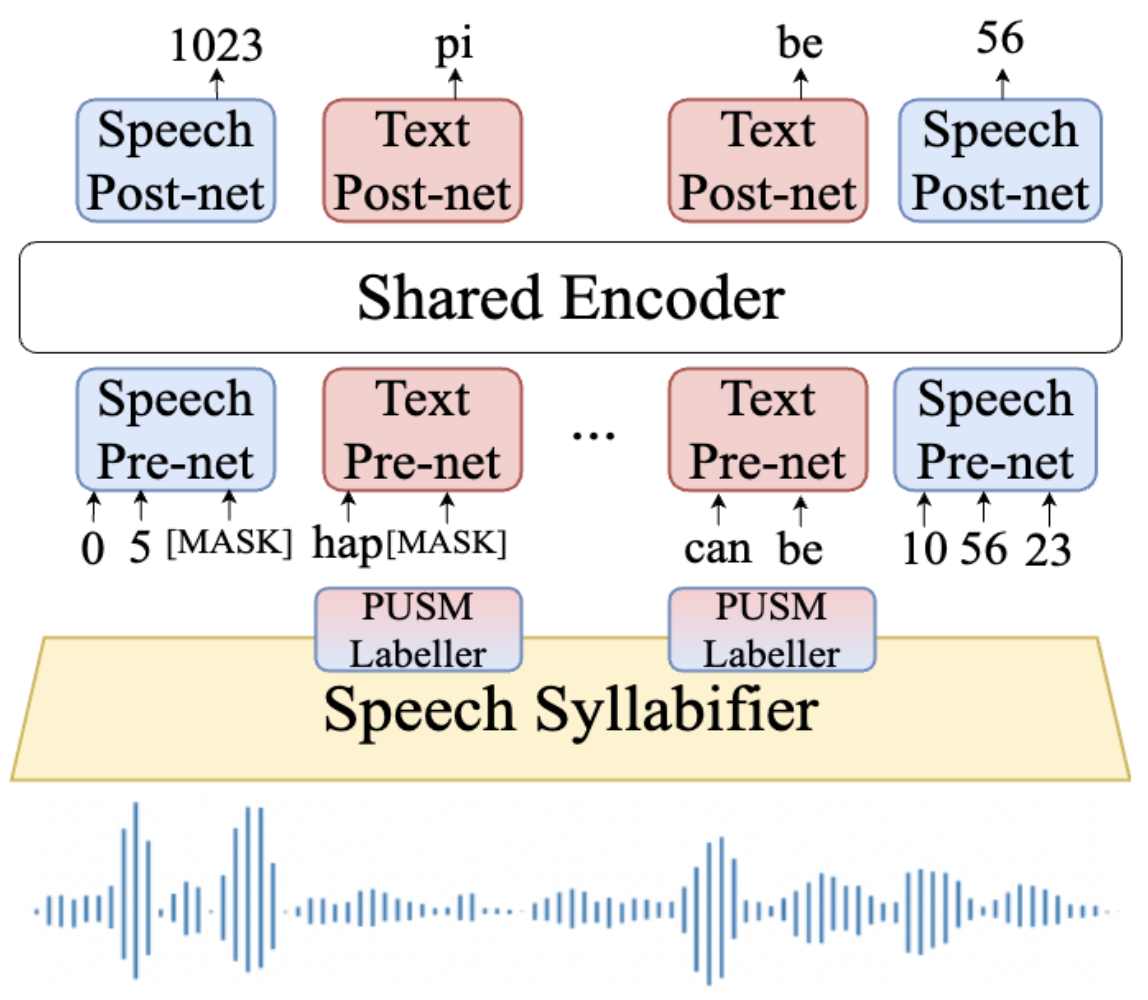}
    \caption{Initialization of \sylcipher/ using interleaving speech tokens and PUSM-predicted text tokens to stabilize training for extreme low-resource scenario.}
    \label{fig:interleave_init}
\end{figure}

\begin{table*}[t]
\centering
\begin{threeparttable}
\small
\setlength{\tabcolsep}{4pt}
\resizebox{1.0\textwidth}{!}{
\begin{tabularx}{\textwidth}{@{}p{0.19\textwidth}Y@{}}
\toprule
Component & Configuration \\
\midrule
Acoustic features
& Layer 21 of HuBERT-Large pretrained on 60k-hour LibriLight or last layer (18) of XEUS encoder\\
\midrule
Syllable-boundary detecter
& Five-layer 1-D CNN/TDNN with stride 1 and same padding. Kernel widths are $[11,9,7,5,3]$; channel dimensions are $1024\!\rightarrow\!500\!\rightarrow\!500\!\rightarrow\!500\!\rightarrow\!500\!\rightarrow\!500$; ReLU follows each convolution. Two linear heads map 500 dimensions to (i) one sigmoid boundary probability and (ii) 100 frame-level acoustic classes. \\
\midrule
Soft-pooler
& soft segment assignment/pooling uses sharpness $\epsilon=0.5$. A fixed $K$-means codebook converts each pooled word representation to a straight-through hard one-hot word token. \\
\midrule
Speech pre-net
& Learned speech-token embedding plus scaled positional encoding; embedding/output dimension 768; no additional hidden layer. The input inventory contains the $K$ acoustic syllable tokens and special symbols. \\
\midrule
Text pre-net
& Learned text-token embedding plus scaled positional encoding; embedding/output dimension 768; no additional hidden layer. The input inventory contains the $K$ written-word types and special symbols. \\
\midrule
Shared encoder
& Two Transformer encoder layers; model dimension 768; feed-forward dimension 3072; 12 attention heads; GELU activation; dropout 0.1; attention dropout 0.1; activation dropout 0; layer-drop 0.05; relative-position embeddings; BERT-style initialization. Maximum speech/text sequence length is 100 tokens. \\
\midrule
Shared mix-up quantizer
& Gumbel--Softmax product quantizer with 3 groups and 100 entries per group. Each encoder position is independently replaced by a shared codebook vector with probability 0.3 in the static-boundary SylCipher setting. \\
\midrule
Speech post-net
& One linear projection from 768 encoder dimensions to speech-token logits ($K$ acoustic types plus special symbols). It is used to reconstruct masked and unmasked speech inputs. \\
\midrule
Text post-net
& One linear projection from 768 encoder dimensions to text-word logits ($K$ word types plus special symbols). It is used for text reconstruction and speech-to-text prediction. \\
\bottomrule
\end{tabularx}}
\caption{\textbf{Architecture used by \sylcipher/.}}
\label{tab:sylcipher-architecture}
\end{threeparttable}
\end{table*}

\begin{table*}[t]
\centering
\begin{threeparttable}
\small
\setlength{\tabcolsep}{4pt}
\resizebox{1.0\textwidth}{!}{
\begin{tabularx}{\textwidth}{@{}p{0.13\textwidth}p{0.23\textwidth}Y@{}}
\toprule
Stage & Supervision & Procedure and output \\
\midrule
Codebook preparation
& Sylber boundaries
& Mean-pool features within Sylber segments; train one
$K$-means model and fix it during training. \\
\midrule
Behavior cloning
& Sylber boundaries
& Train the five-layer CNN with boundary BCE (positive-class weight 1.5) and
frame-level acoustic cross-entropy. The acoustic targets are 100-means labels
from HuBERT-base layer 6. Both loss weights are 1. \\
\bottomrule
\end{tabularx}}
\caption{\textbf{Boundary and codebook initialization}.}
\label{tab:boundary-refinement}
\end{threeparttable}
\end{table*}

\begin{table*}[t]
\centering
\begin{threeparttable}
\small
\setlength{\tabcolsep}{4pt}
\resizebox{1.0\textwidth}{!}{
\begin{tabularx}{\textwidth}{@{}p{0.17\textwidth}p{0.38\textwidth}Y@{}}
\toprule
Experimental setting & Optimization & Schedule, batching, and selection \\
\midrule
Fixed-boundary MLM-based training & Adam; learning rate $2\times 10^{-4}$; $\beta =(0.9,0.98)$;
$\epsilon =10^{-6}$; weight decay 0; gradient-norm clipping 20; polynomial decay
(power 1); 150 warm-up updates; maximum/total updates $10^{6}$. Poisson span masking with mean length 3.5 and a mask budget below 30\%.
Replace 90\% of selected positions by \texttt{<MASK>} and 10\% by a random token.
Masked-token loss weight 1; unmasked-token loss weight 0.5; speech and text
modalities weighted equally.
& Eight 32-GB V100 GPUs; at most 500 tokens per GPU; update frequency 16;
speech:text batch ratio 1:1; seed 0. Random speech and text utterances are paired in each training batch. Shared
codebook mix-up probability is 0.3. Train the fixed-boundary stage from scratch
for fewer than 7000 epochs; initialize later JE2E stage from the fixed-boundary checkpoint and train for fewer than 1600 epochs and PUSM stage for 2000 epochs (or 2000 steps by using the whole dataset as a single batch). Stop when paired
validation SER flattens. \\
\midrule
JE2E boundary refinement
& Initialize from converged static-boundary \sylcipher/ and CNN checkpoints. Adam with
learning rate $2\times 10^{-4}$; text-loss ratio 10; end-to-end policy-loss weight
1 and policy-pretraining weight 0.
& Update frequency 890, chosen to produce approximately one optimizer update per
epoch for the 2048-word setup. Stop before
700 epochs when validation WER converges.\\
\midrule
PUSM training
& Optimize the PUSM objective at learning rate $10^{-3}$.
& Use an update of 7200 to approximately use the whole dataset as a single batch; Codebook mix-up probability is set to 0.\\
\midrule
Pseudo-text CTC student
& Initialize the unfinetuned 60k-hour wav2vec 2.0 checkpoint. Adam;
learning rate $10^{-4}$; $\beta =(0.9,0.98)$; $\epsilon =10^{-8}$; 
& One GPU; a batch size of 32 with a one-layer bi-directional LSTM classifier network with 512 hidden units. Select the checkpoint by
validation WER. \\
\bottomrule
\end{tabularx}}
\caption{\textbf{Optimization parameters and training schedule.}}
\label{tab:optimization}
\end{threeparttable}
\end{table*}

\begin{table}[ht]
    \centering
    \begin{tabular}{l|c}
        \toprule
        \textbf{Seed} & CER ($\downarrow$) \\
        \midrule
        0 & 33.0 \\
        1 & 33.4 \\
        2 & 33.2 \\
        3 & 34.1 \\
        mean$\pm$std & 33.4$\pm$0.5 \\
        \bottomrule
    \end{tabular}
    \caption{\textbf{Effect of random seeds.} The model is a syllable-level PUSM with XEUS backbone trained on LibriSpeech.}
    \label{tab:eff_random_seed}
\end{table}

\begin{table}[ht]
    \centering
    \begin{tabular}{l|ccc}
        \toprule
        \textbf{Model} & $\lambda_{\text{cp}}$ & $\lambda_{\text{gp}}$ & \textbf{SER ($\downarrow$)}\\
        \midrule
        Syllabic GAN & 2 & 1.5 & 115.8\\
        Syllabic GAN & 2 & 2 & 115.7\\
        Syllabic GAN & 4 & 1.5 & 123.7\\
        Syllabic GAN & 4 & 2 & 126.5\\
        SylCipher & - & - & \textbf{25.5}\\
        \bottomrule
    \end{tabular}
    \caption{\textbf{Comparison with syllable-level GAN model on LibriSpeech}. The GAN uses the same architecture and training objective as \wtovu/. $\lambda_{\text{cp}}$ stands for the code penalty regularization loss and $\lambda_{\text{gp}}$ stands for the gradient penalty loss, both used by \wtovu/. SER stands for syllable error rate (lower is better).}
    \label{tab:comparison_with_syllable_gan}
\end{table}

\section{Implementation details of \sylcipher/}\label{app:implementation}
For English experiments, \sylcipher/ uses a \hubl/-large\footnote{\tiny https://dl.fbaipublicfiles.com/hubert/hubert\_large\_ll60k.pt} model pretrained on LibriLight~\citep{librilight} as the SSL encoder in the speech syllabifier, and initialize the soft-pooler with unsupervised boundary labels from Sylber~\citep{Cho2025-sylber}. For languages other than English, we instead use XEUS\footnote{\tiny https://huggingface.co/espnet/xeus/blob/main/model/xeus\_checkpoint\_new.pth}~\citep{chen2024-xeus}, which is pretrained on Mandarin (among other languages) and significantly outperforms \hubl/. We also finetune Sylber on \aishell/ to ensure stable convergence. For the \prenet/s, shared encoder, \postnet/s, MLM parameters, and most of the optimizer hyperparameters, we follow JSTTI~\citep{Ni2025-jstti}, as modifying them did not yield consistent improvements. More implementation details can be found in Table~\ref{tab:sylcipher-architecture},~\ref{tab:boundary-refinement} and \ref{tab:optimization}.

\section{Additional experiments: multiple random seeds and comparison with GAN}
Table~\ref{tab:eff_random_seed} evaluates the robustness of the syllable-level PUSM to random initialization. Across four random seeds, the model obtains CERs between 33.0 and 34.1, with a mean of 33.4$\pm$0.5. The small standard deviation and narrow performance range indicate that the proposed training procedure is stable and that its performance is not sensitive to the choice of random seed.

Table~\ref{tab:comparison_with_syllable_gan} compares SylCipher with a syllable-level GAN based on the \wtovu/ architecture and objective. The GAN performs poorly across all tested combinations of code-penalty and gradient-penalty weights, yielding SERs between 115.7 and 126.5. Increasing the code-penalty weight from 2 to 4 further degrades performance, while varying the gradient-penalty weight has a comparatively smaller effect. In contrast, SylCipher achieves an SER of 25.5, corresponding to a 78.0\% relative error reduction over the best GAN configuration. These results suggest that directly extending the adversarial wav2vec-U framework to syllable-level recognition is ineffective, whereas SylCipher provides substantially more reliable syllable-level transcription.

When training \sylcipher/ on Taigi, an extreme low-resource language, we use an interleaving language modeling strategy shown in Figure~\ref{fig:interleave_init} to initialize the model. First, we train an PUSM model and use it to generate pseudo-text labels for the utterances. Then we uniformly randomly interleave the PUSM text tokens with the speech tokens from the speech syllabifier with an interleaving probability of 0.5. We found that this approach allows the model to learn a more aligned shared latent space for speech and text than naive cross-entropy training.

\section{Pseudo-code for the Pyphen+ syllabifier}\label{app:code_pyphen+}
See Algorithm~\ref{alg:code_pyphen+}.

\section{Pseudo-code of BPE+ syllabifier}\label{app:code_bpe+}
See Algorithm~\ref{alg:code_bpe+}.

\begin{algorithm}[ht]
\footnotesize
\caption{Merging No-Vowel Segments and Naive Syllabification}
\begin{algorithmic}[1]
\Function{merge\_novowel\_segments}{segs}
    \State new\_segs $\gets$ [ ]
    \State n\_seg $\gets$ len(segs)
    \State i $\gets 0$
    \While{$i < n\_seg$}
        \State seg $\gets$ segs[i]
        \If{seg has no vowels}
            \If{seg is the last segment}
                \State merge seg into previous segment (if any)
            \Else
                \State merge seg into the next segment
                \State $i \gets i + 1$
            \EndIf
        \Else
            \State append seg to new\_segs
        \EndIf
        \State $i \gets i + 1$
    \EndWhile
    \State \Return new\_segs
\EndFunction

\Function{naive\_syllabify}{w}
    \State w0 $\gets$ w
    \If{w ends with a silent `e' (not ``-le'')}
        \State drop the final `e' in w
    \EndIf
    \If{w ends with ``-ed'' and root does not end with `t' or `d'}
        \State drop the `e' in ``-ed'' in w
    \EndIf

    \State Apply rule to w: split \texttt{V C-CC… V} into \texttt{VC-CC…V}
    \State Apply rule to w: split \texttt{VCCV} into \texttt{VC-CV}
    \State Apply rule to w: split \texttt{VCV} into \texttt{V-CV}

    \If{modifications were made}
        \If{w0 ended with `e'}
            \State add back `e' to w
        \ElsIf{w0 ended with `ed'}
            \State restore `ed' to w
        \EndIf
    \EndIf
    \State \Return syllabified w
\EndFunction

\Function{pyphen\_plus\_syllabify}{w}
    \State syls $\gets$ []
    \State segs $\gets$ \href{https://github.com/grantjenks/python-wordsegment}{wordsegment}.segment(w)
    \For{each seg in segs}
        \State syls\_ $\gets$ pyphen.inserted(seg)
        \State syls\_ $\gets$ merge\_novowel\_segments(syls\_)
        \If{syls\_ has only one syllable \textbf{and} \href{https://github.com/mholtzscher/syllapy}{syllapy}.count(syls\_) $>$ 1} 
            \State syls\_ $\gets$ naive\_syllabify(syls\_)
        \EndIf
        \State extend syls with syls\_ 
    \EndFor
\EndFunction
\end{algorithmic}
\label{alg:code_pyphen+}
\end{algorithm}

\begin{algorithm}[ht]
\footnotesize
\caption{Splitting and Processing BPE Tokens with Vowel Constraints}
\begin{algorithmic}[1]
\Function{split\_on\_nonconsecutive\_vowels}{token}
    \State parts $\gets$ [ ]
    \State current $\gets$ [ ]
    \State vowel\_positions $\gets$ [ ]
    \For{each character $ch$ in token}
        \State append $ch$ to current
        \If{$ch \in$ VOWELS}
            \State record position of $ch$ in current
        \EndIf
        \If{two or more vowels are non-consecutive}
            \State cut before the last vowel
            \State append left substring to parts
            \State reset current and vowel\_positions accordingly
        \EndIf
    \EndFor
    \If{current is not empty}
        \State append current to parts
    \EndIf
    \State \Return parts
\EndFunction

\Function{enforce\_vowel\_constraint}{parts}
    \State remove empty parts
    \State new\_parts $\gets$ [ ]
    \State buffer $\gets$ empty string
    \For{each part $p$ in parts}
        \State buffer $\gets$ buffer + $p$
        \If{buffer contains a vowel}
            \If{next part starts with a vowel \textbf{or} ``E''}
                \State continue without breaking
            \Else
                \State append buffer to new\_parts
                \State reset buffer
            \EndIf
        \EndIf
    \EndFor
    \If{buffer not empty}
        \If{new\_parts not empty}
            \State merge buffer into the last part
        \Else
            \State append buffer as a new part
        \EndIf
    \EndIf
    \State \Return new\_parts
\EndFunction

\Function{bpe\_plus\_syllabify}{tokens}
    \State first\_split $\gets$ apply split\_on\_nonconsecutive\_vowels to each token
    \State final\_parts $\gets$ enforce\_vowel\_constraint(first\_split)
    \State \Return final\_parts
\EndFunction
\end{algorithmic}
\label{alg:code_bpe+}
\end{algorithm}

\section{Spectrogram examples on LibriSpeech}\label{app:spec_libri}
\begin{figure*}
    \centering
    \begin{subfigure}{0.9\textwidth}
    \includegraphics[width=0.99\textwidth]{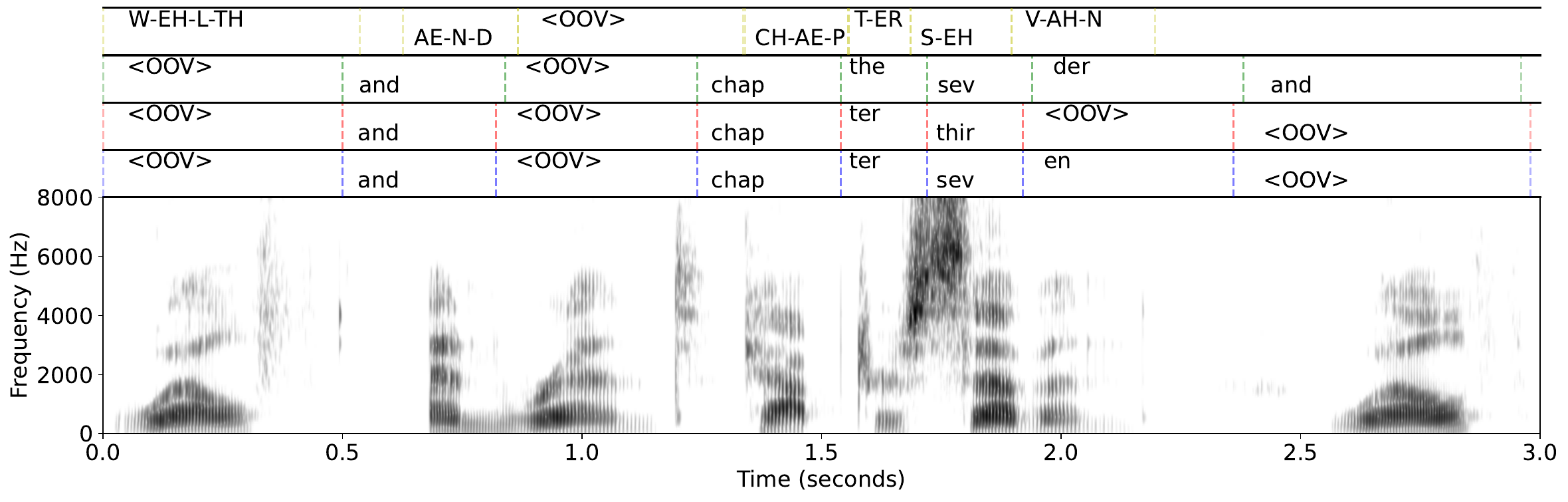}
    \caption{Reference transcript: ``Wealth and rent chap-ter sev-en wealth''.}    
    \end{subfigure}
    \begin{subfigure}{0.9\textwidth}
    \includegraphics[width=0.99\textwidth]{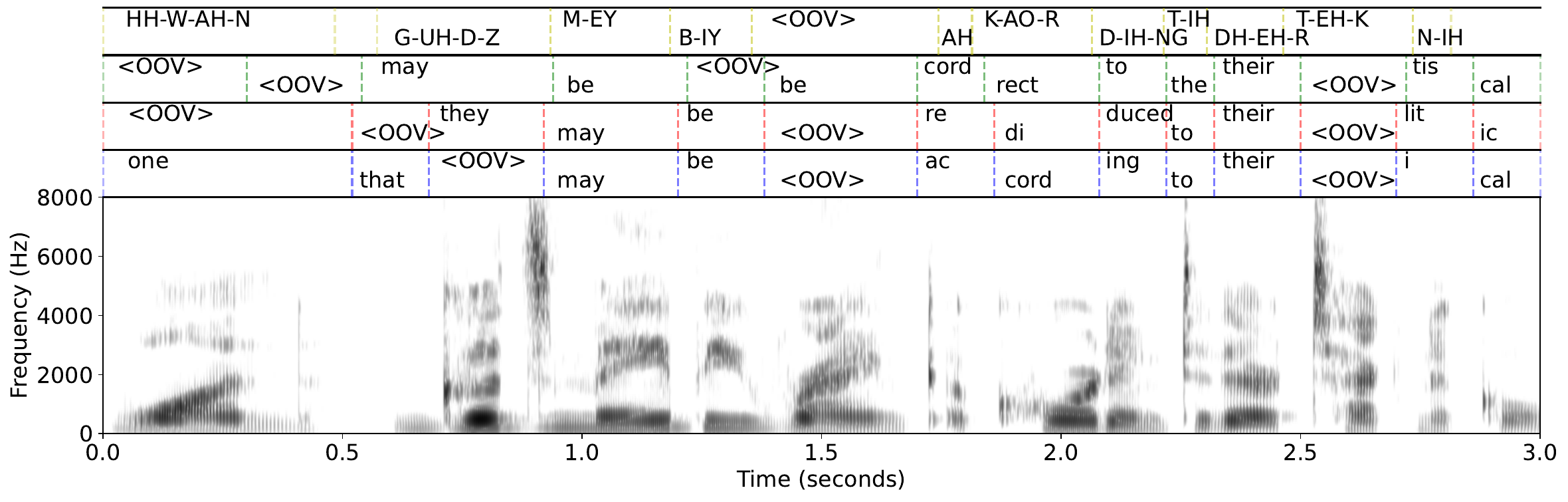}
    \caption{Reference transcript: ``One goods may be ranked ac-cord-ing to their tech-ni-cal''.}    
    \end{subfigure}
    \begin{subfigure}{0.9\textwidth}
    \includegraphics[width=0.99\textwidth]{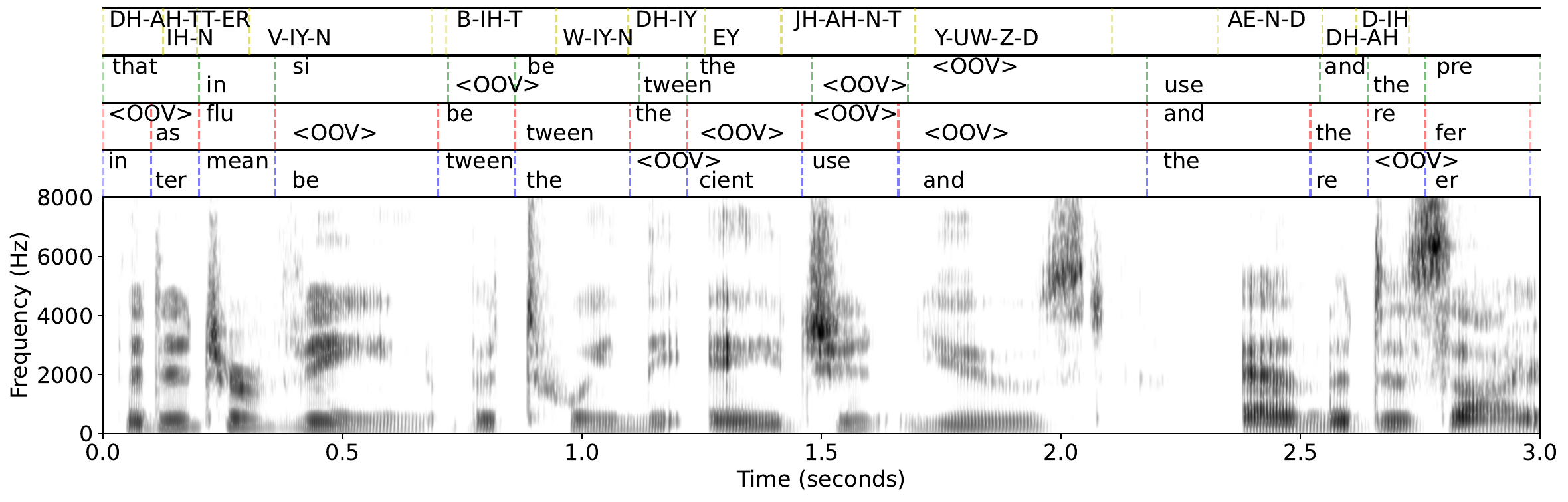}
    \caption{Reference transcript: ``That in-ter-vene be-tween the a-gent used and the de-sired''}    
    \end{subfigure}
    \begin{subfigure}{0.9\textwidth}
    \includegraphics[width=0.99\textwidth]{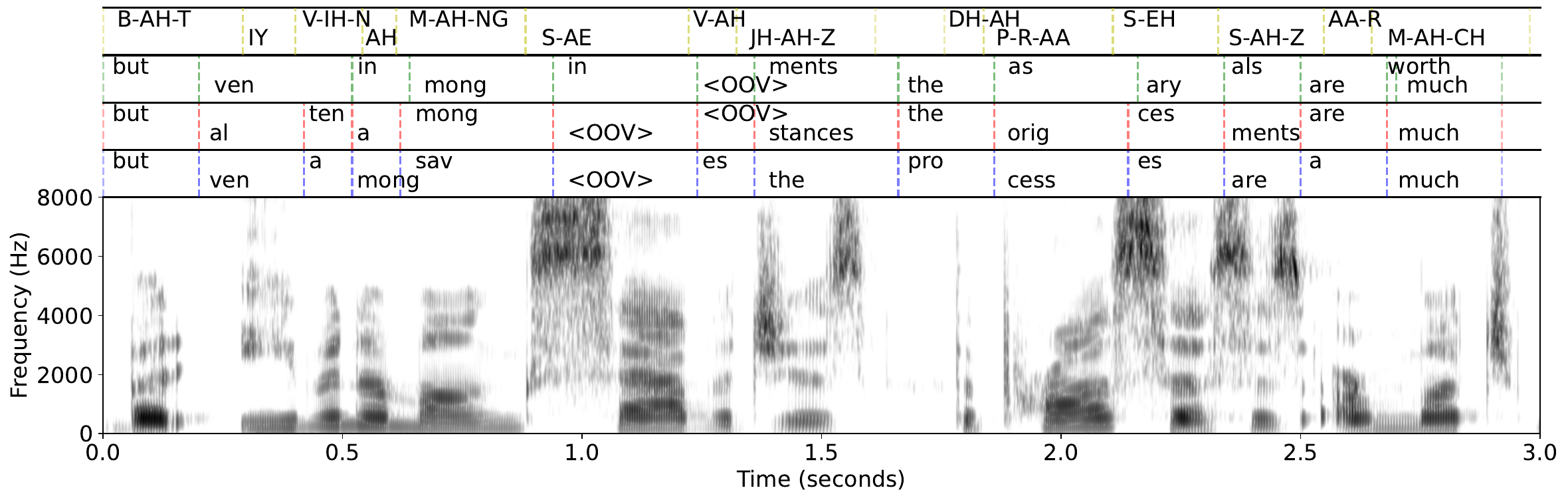}
    \caption{Reference transcript: ``But e-ven a-mong sav-ages the pro-cess-es are much''}    
    \end{subfigure}
    \caption{\textbf{Spectrograms of audio examples in our test split of LibriSpeech clean subsets (matched setting) and the predicted speech-text alignment by \sylcipher/ after different training stages.} Audios are truncated to the 3-second mark for better visualization. The alignment bars from top to bottom: Forced alignment, Sylber, Sylber+JE2E, Sylber+JE2E+PUSM.}
\end{figure*}

\begin{figure*}
    \centering
    \begin{subfigure}{0.9\textwidth}
    \includegraphics[width=0.99\textwidth]{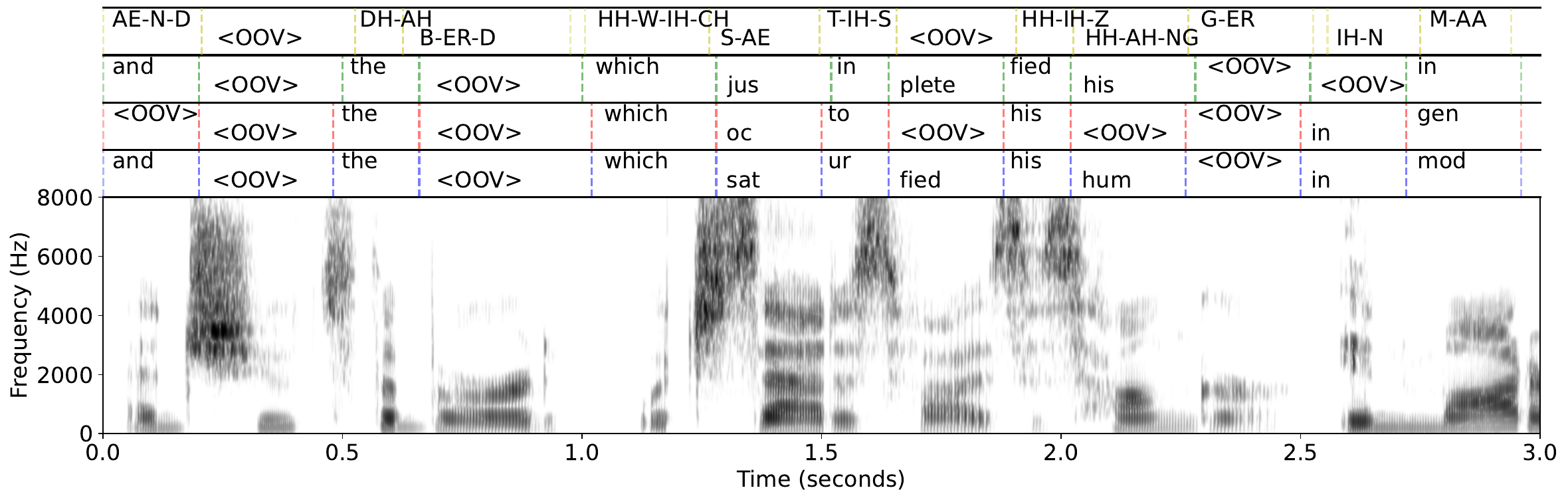}
    \caption{Reference transcript: ``And shoots the bird which sat-is-flies his hun-ger''.}    
    \end{subfigure}
    \begin{subfigure}{0.9\textwidth}
    \includegraphics[width=0.99\textwidth]{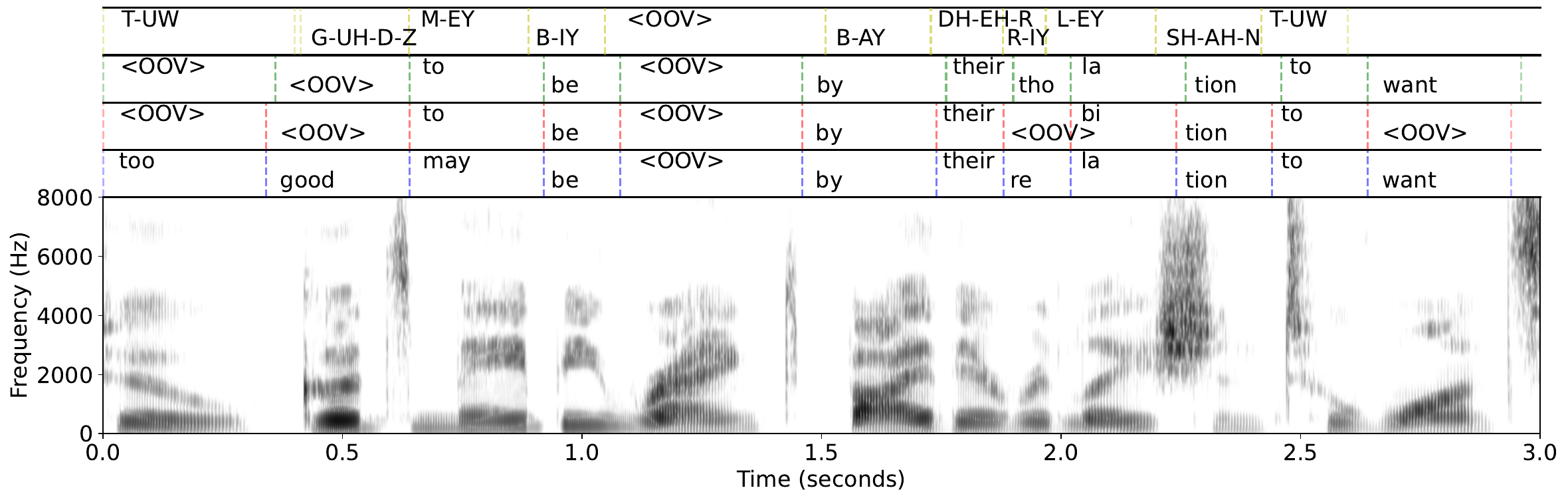}
    \caption{Reference transcript: ``two goods may be ranked by their re-la-tion to wants''.}    
    \end{subfigure}
    \begin{subfigure}{0.9\textwidth}
    \includegraphics[width=0.99\textwidth]{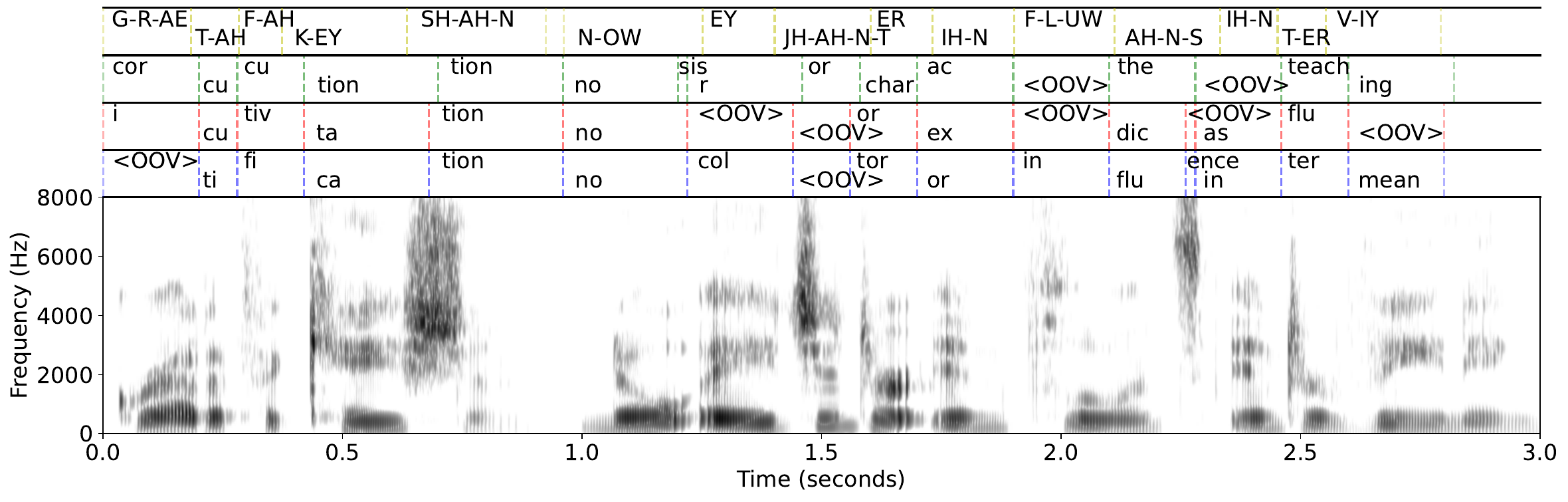}
    \caption{Reference transcript: ``Grat-i-fi-ca-tion no a-gent or in-flu-ence in-ter-ven-ing''.}    
    \end{subfigure}
    \begin{subfigure}{0.9\textwidth}
    \includegraphics[width=0.99\textwidth]{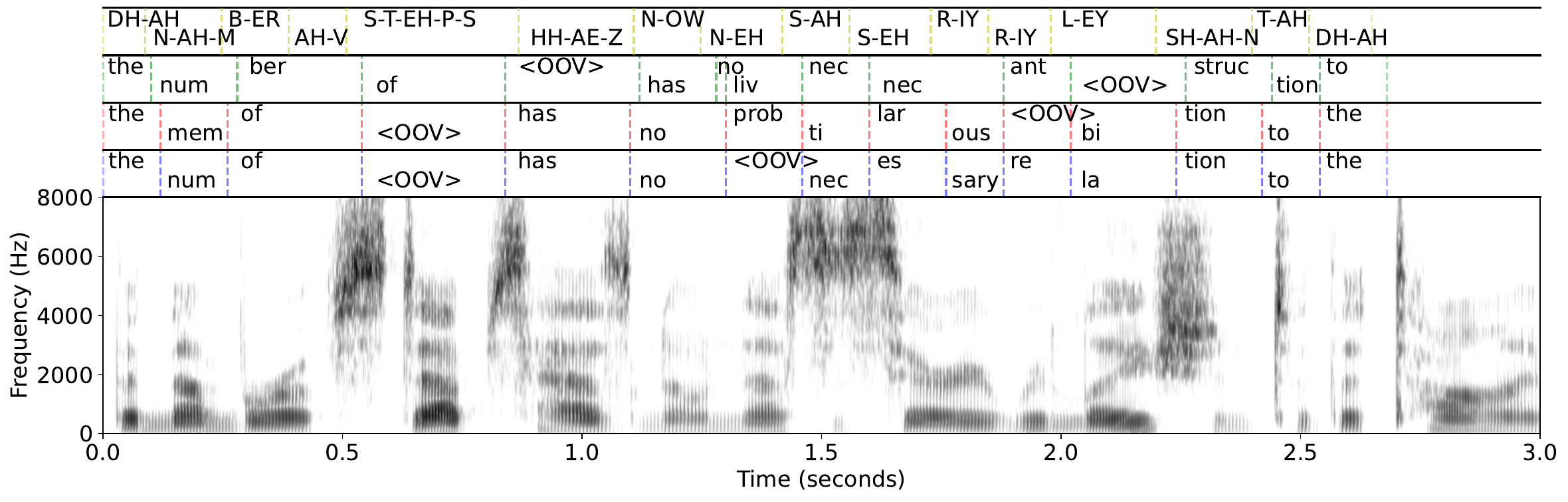}
    \caption{Reference transcript: ``The num-ber of steps has no nec-es-sary re-la-tion to the''.}    
    \end{subfigure}
    \caption{\textbf{Spectrograms of audio examples in our test split of LibriSpeech clean subsets (matched setting) and the predicted speech-text alignment by \sylcipher/ after different training stages.} Audios are truncated to the 3-second mark for better visualization. The alignment bars from top to bottom: Forced alignment, Sylber, Sylber+JE2E, Sylber+JE2E+PUSM.}
\end{figure*}

%% file: emnlp2026.bib
@string {nips = "Neural Information Processing Systems"}

@string{ICLR = "ICLR"}

@string{icassp = "ICASSP"}

@string {interspeech = "Interspeech"}

@string{TMLR = "TMLR"}

@inproceedings{Ao2022-speecht5,
  author       = {Junyi Ao and
                  Rui Wang and
                  Long Zhou and
                  Chengyi Wang and
                  Shuo Ren and
                  Yu Wu and
                  Shujie Liu and
                  Tom Ko and
                  Qing Li and
                  Yu Zhang and
                  Zhihua Wei and
                  Yao Qian and
                  Jinyu Li and
                  Furu Wei},
  editor       = {Smaranda Muresan and
                  Preslav Nakov and
                  Aline Villavicencio},
  title        = {{SpeechT5}: Unified-Modal Encoder-Decoder Pre-Training for Spoken Language
                  Processing},
  booktitle    = {Proceedings of the 60th Annual Meeting of the Association for Computational
                  Linguistics (Volume 1: Long Papers), {ACL} 2022, Dublin, Ireland,
                  May 22-27, 2022},
  pages        = {5723--5738},
  publisher    = {Association for Computational Linguistics},
  year         = {2022},
  url          = {https://doi.org/10.18653/v1/2022.acl-long.393},
  doi          = {10.18653/V1/2022.ACL-LONG.393},
}

@inproceedings{Baade2025-syllablelm,
  title={SyllableLM: Learning Coarse Semantic Units for Speech Language Models},
  author={Baade, Alan and Peng, Puyuan and Harwath, David},
  booktitle=iclr,
  year={2025}
}

@article{Bhati2022-scpc,
  author       = {Saurabhchand Bhati and
                  Jes{\'{u}}s Villalba and
                  Piotr Zelasko and
                  Laureano Moro{-}Vel{\'{a}}zquez and
                  Najim Dehak},
  title        = {Unsupervised Speech Segmentation and Variable Rate Representation
                  Learning Using Segmental Contrastive Predictive Coding},
  journal      = {{IEEE} {ACM} Trans. Audio Speech Lang. Process.},
  volume       = {30},
  pages        = {2002--2014},
  year         = {2022},
  url          = {https://doi.org/10.1109/TASLP.2022.3180684},
  doi          = {10.1109/TASLP.2022.3180684},
}

@inproceedings{Chen2019-uasr,
  author       = {Kuan{-}Yu Chen and
                  Che{-}Ping Tsai and
                  Da{-}Rong Liu and
                  Hung{-}yi Lee and
                  Lin{-}Shan Lee},
  editor       = {Gernot Kubin and
                  Zdravko Kacic},
  title        = {Completely Unsupervised Phoneme Recognition by a Generative Adversarial
                  Network Harmonized with Iteratively Refined Hidden Markov Models},
  booktitle    = {Interspeech 2019, 20th Annual Conference of the International Speech
                  Communication Association, Graz, Austria, 15-19 September 2019},
  pages        = {1856--1860},
  publisher    = {{ISCA}},
  year         = {2019},
  url          = {https://doi.org/10.21437/Interspeech.2019-2068},
  doi          = {10.21437/INTERSPEECH.2019-2068},
}

@article{arora2025-slm,
  title={On The Landscape of Spoken Language Models: A Comprehensive Survey},
  author={Arora, Siddhant and Chang, Kai-Wei and Chien, Chung-Ming and Peng, Yifan and Wu, Haibin and Adi, Yossi and Dupoux, Emmanuel and Lee, Hung-Yi and Livescu, Karen and Watanabe, Shinji},
  journal=TMLR,
  year={2025}
}

@inproceedings{panayotov2015librispeech,
  title={Librispeech: An ASR corpus based on public domain audio books},
  author={Panayotov, Vassil and Chen, Guoguo and Povey, Daniel and Khudanpur, Sanjeev},
  booktitle={2015 IEEE International Conference on Acoustics, Speech and Signal Processing (ICASSP)},
  pages={5206--5210},
  year={2015},
  organization={IEEE},
  doi={10.1109/ICASSP.2015.7178964}
}

@inproceedings{Artetxe18-unsupnmt,
  author       = {Mikel Artetxe and
                  Gorka Labaka and
                  Eneko Agirre and
                  Kyunghyun Cho},
  title        = {Unsupervised Neural Machine Translation},
  booktitle    = {6th International Conference on Learning Representations, {ICLR} 2018,
                  Vancouver, BC, Canada, April 30 - May 3, 2018, Conference Track Proceedings},
  publisher    = {OpenReview.net},
  year         = {2018},
  url          = {https://openreview.net/forum?id=Sy2ogebAW},
}

@inproceedings{Baevski2021-wav2vec-u,
  author       = {Alexei Baevski and
                  Wei{-}Ning Hsu and
                  Alexis Conneau and
                  Michael Auli},
  editor       = {Marc'Aurelio Ranzato and
                  Alina Beygelzimer and
                  Yann N. Dauphin and
                  Percy Liang and
                  Jennifer Wortman Vaughan},
  title        = {Unsupervised Speech Recognition},
  booktitle    = {Advances in Neural Information Processing Systems 34: Annual Conference
                  on Neural Information Processing Systems 2021, NeurIPS 2021, December
                  6-14, 2021, virtual},
  pages        = {27826--27839},
  year         = {2021},
  url = {https://papers.nips.cc/paper_files/paper/2021/hash/ea159dc9788ffac311592613b7f71fbb-Abstract.html},
}

@inproceedings{Baevski2020-wav2vec2,
  author       = {Alexei Baevski and
                  Yuhao Zhou and
                  Abdelrahman Mohamed and
                  Michael Auli},
  editor       = {Hugo Larochelle and
                  Marc'Aurelio Ranzato and
                  Raia Hadsell and
                  Maria{-}Florina Balcan and
                  Hsuan{-}Tien Lin},
  title        = {wav2vec 2.0: {A} Framework for Self-Supervised Learning of Speech
                  Representations},
  booktitle    = {Advances in Neural Information Processing Systems 33: Annual Conference
                  on Neural Information Processing Systems 2020, NeurIPS 2020, December
                  6-12, 2020, virtual},
 pages = {12449--12460},
 year = {2020},
  url          = {https://proceedings.neurips.cc/paper/2020/hash/92d1e1eb1cd6f9fba3227870bb6d7f07-Abstract.html},
}

@article{chen2022-wavlm,
  author       = {Sanyuan Chen and
                  Chengyi Wang and
                  Zhengyang Chen and
                  Yu Wu and
                  Shujie Liu and
                  Zhuo Chen and
                  Jinyu Li and
                  Naoyuki Kanda and
                  Takuya Yoshioka and
                  Xiong Xiao and
                  Jian Wu and
                  Long Zhou and
                  Shuo Ren and
                  Yanmin Qian and
                  Yao Qian and
                  Jian Wu and
                  Michael Zeng and
                  Xiangzhan Yu and
                  Furu Wei},
  title        = {{WavLM}: Large-Scale Self-Supervised Pre-Training for Full Stack Speech
                  Processing},
  journal      = {{IEEE} J. Sel. Top. Signal Process.},
  volume       = {16},
  number       = {6},
  pages        = {1505--1518},
  year         = {2022},
  url          = {https://doi.org/10.1109/JSTSP.2022.3188113},
  doi          = {10.1109/JSTSP.2022.3188113},
}

@inproceedings{chen2024-xeus,
  author       = {William Chen and Wangyou Zhang and Yifan Peng and Xinjian Li and Jinchuan Tian and Jiatong Shi and Xuankai Chang and Soumi Maiti and Karen Livescu and Shinji Watanabe},
  title        = {Towards Robust Speech Representation Learning for Thousands of Languages},
  booktitle    = {Proceedings of the 2024 Conference on Empirical Methods in Natural Language Processing (EMNLP)},
  pages        = {10205--10224},
  year         = {2024},
  publisher    = {Association for Computational Linguistics},
  doi          = {10.48550/arXiv.2407.00837},
  note         = {Preprint available on arXiv: arXiv:2407.00837}
}

@inproceedings{cho2023-sdhubert,
  title={SD-HuBERT: Sentence-Level Self-Distillation Induces Syllabic Organization in HuBERT},
  author={Cho, Cheol Jun and Mohamed, Abdelrahman and Li, Shang-Wen and Black, Alan W and Anumanchipalli, Gopala K},
  booktitle={ICASSP},
  year={2024}
}

@inproceedings{Cho2025-sylber,
  title     = {Sylber: Syllabic Embedding Representation of Speech from Raw Audio},
  author    = {Cheol Jun Cho and Nicholas Lee and Akshat Gupta and Dhruv Agarwal and Ethan Chen and Alan W. Black and Gopala K. Anumanchipalli},
  booktitle = iclr,
  year      = {2025},
  address   = {Singapore}
}

@inproceedings{chung2021w2vbert,
  author    = {Yu-An Chung and Yu Zhang and Wei Han and Chung-Cheng Chiu and James Qin and Ruoming Pang and Yonghui Wu},
  title     = {{W2v-BERT}: Combining Contrastive Learning and Masked Language Modeling for Self-Supervised Speech Pre-Training},
  booktitle = {2021 IEEE Automatic Speech Recognition and Understanding Workshop (ASRU)},
  pages     = {244--250},
  year      = {2021},
  publisher = {IEEE}
}

@inproceedings{fuchs2023-gradseg,
  author       = {Tzeviya Sylvia Fuchs and
                  Yedid Hoshen},
  title        = {Unsupervised Word Segmentation Using Temporal Gradient Pseudo-Labels},
  booktitle    = {{IEEE} International Conference on Acoustics, Speech and Signal Processing
                  {ICASSP} 2023, Rhodes Island, Greece, June 4-10, 2023},
  pages        = {1--5},
  publisher    = {{IEEE}},
  year         = {2023},
  url          = {https://doi.org/10.1109/ICASSP49357.2023.10095363},
  doi          = {10.1109/ICASSP49357.2023.10095363},
}

@inproceedings{Glass2012-unsup-speech,
    author = {James Glass},
    title = {Towards Unsupervised Speech Processing},
    booktitle = {International Conference on Information Sciences, Signal Processing and their Applications},
    year = {2012},
    url = {https://groups.csail.mit.edu/sls/publications/2012/Glass-ISSPA12.pdf},
}

@inproceedings{goldwasser2023theory,
  title     = {A Theory of Unsupervised Translation Motivated by Understanding Animal Communication},
  author    = {Shafi Goldwasser and David F. Gruber and Adam Tauman Kalai and Orr Paradise},
  booktitle = nips,
  year      = {2023}
}

@inproceedings{hsu2021textfree,
  author    = {Hsu, Wei-Ning and Harwath, David and Song, Chunxi and Glass, James},
  title     = {Text-Free Image-to-Speech Synthesis Using Learned Segmental Units},
  booktitle = {Proceedings of the 59th Annual Meeting of the Association for Computational Linguistics (ACL)},
  year      = {2021},
  pages     = {5820--5831},
  doi       = {10.18653/v1/2021.acl-long.454}
}

@article{Hsu2022-hubert,
  author       = {Wei{-}Ning Hsu and
                  Benjamin Bolte and
                  Yao{-}Hung Hubert Tsai and
                  Kushal Lakhotia and
                  Ruslan Salakhutdinov and
                  Abdelrahman Mohamed},
  title        = {{HuBERT}: Self-Supervised Speech Representation Learning by Masked Prediction
                  of Hidden Units},
  journal      = {{IEEE} {ACM} Trans. Audio Speech Lang. Process.},
  volume       = {29},
  pages        = {3451--3460},
  year         = {2021},
  url          = {https://doi.org/10.1109/TASLP.2021.3122291},
  doi          = {10.1109/TASLP.2021.3122291},
}

@inproceedings{Lample2018-unsupmtmono,
  author       = {Guillaume Lample and
                  Alexis Conneau and
                  Ludovic Denoyer and
                  Marc'Aurelio Ranzato},
  title        = {Unsupervised Machine Translation Using Monolingual Corpora Only},
  booktitle    = {6th International Conference on Learning Representations, {ICLR} 2018,
                  Vancouver, BC, Canada, April 30 - May 3, 2018, Conference Track Proceedings},
  publisher    = {OpenReview.net},
  year         = {2018},
  url          = {https://openreview.net/forum?id=rkYTTf-AZ},
}

@inproceedings{levy2025unsupervised,
  title     = {Unsupervised Translation of Emergent Communication},
  author    = {Ido Levy and Orr Paradise and Boaz Carmeli and Ron Meir and Shafi Goldwasser and Yonatan Belinkov},
  booktitle = {Proceedings of the Thirty-Ninth AAAI Conference on Artificial Intelligence (AAAI-25)},
  year      = {2025},
  address   = {Philadelphia, PA, USA},
  publisher = {AAAI Press}
}

@INPROCEEDINGS{librilight,
  author={J. {Kahn} and M. {Rivière} and W. {Zheng} and E. {Kharitonov} and Q. {Xu} and P. E. {Mazaré} and J. {Karadayi} and V. {Liptchinsky} and R. {Collobert} and C. {Fuegen} and T. {Likhomanenko} and G. {Synnaeve} and A. {Joulin} and A. {Mohamed} and E. {Dupoux}},
  booktitle=icassp,
  title={Libri-Light: A Benchmark for ASR with Limited or No Supervision}, 
  year={2020},
  pages={7669-7673}
}

@inproceedings{Liu2018-asru,
    author = {Da-Rong Liu and Kuan-Yu Chen and Hung-Yi Lee and Lin-shan Lee},
    title = {Completely unsupervised phoneme recognition by adversarially learning mapping relationships from audio embeddings},
    booktitle=interspeech,
    year={2018},
    url = {https://www.isca-speech.org/archive_v0/Interspeech_2018/pdfs/1800.pdf},
}

@inproceedings{Liu2022-utts,
  title     = {Simple and Effective Unsupervised Speech Synthesis},
  author    = {Alexander H. Liu and Cheng-I Jeff Lai and Wei-Ning Hsu and Michael Auli and Alexei Baevski and James Glass},
  booktitle = {Interspeech 2022},
  pages     = {843--847},
  year      = {2022},
  doi       = {10.21437/Interspeech.2022-11071},
  publisher = {ISCA},
}

@inproceedings{Liu2023-wav2vecu2,
  author={Liu, Alexander H. and Hsu, Wei-Ning and Auli, Michael and Baevski, Alexei},
  booktitle={2022 IEEE Spoken Language Technology Workshop (SLT)}, 
  title={Towards End-to-End Unsupervised Speech Recognition}, 
  year={2023},
  pages={221--228},
  doi={10.1109/SLT54892.2023.10023187}}

@inproceedings{liu-etal-2025-superbpe,
  title={{SuperBPE}: Space travel for language models},
  author={Alisa Liu and Jonathan Hayase and Valentin Hofmann and Sewoong Oh and Noah A Smith and Yejin Choi},
  booktitle={Second Conference on Language Modeling},
  year={2025},
  url={https://arxiv.org/abs/2503.13423}
}

@inproceedings{ma2019-unpaired,
  title     = {Unpaired Image-to-Speech Synthesis with Multimodal Information Bottleneck},
  author    = {Shuang Ma and Daniel McDuff and Yale Song},
  booktitle = {Proceedings of the IEEE/CVF International Conference on Computer Vision (ICCV)},
  pages     = {7515--7524},
  year      = {2019},
  doi       = {10.1109/ICCV.2019.00765},
  url       = {https://openaccess.thecvf.com/content_ICCV_2019/html/Ma_Unpaired_Image-to-Speech_Synthesis_With_Multimodal_Information_Bottleneck_ICCV_2019_paper.html}
}

@article{mohamed2022self,
  author    = {Abdelrahman Mohamed and Hung-yi Lee and Lasse Borgholt and Jakob D. Havtorn and Joakim Edin and Christian Igel and Katrin Kirchhoff and Shang-Wen Li and Karen Livescu and Lars Maal{\o}e and others},
  title     = {Self-Supervised Speech Representation Learning: A Review},
  journal   = {IEEE Journal of Selected Topics in Signal Processing},
  volume    = {16},
  number    = {6},
  pages     = {1179--1210},
  year      = {2022},
  publisher = {IEEE},
  doi       = {10.1109/JSTSP.2022.3203489}
}

@inproceedings{Ni-unsuptts-interspeech2022,
  author       = {Junrui Ni and
                  Liming Wang and
                  Heting Gao and
                  Kaizhi Qian and
                  Yang Zhang and
                  Shiyu Chang and
                  Mark Hasegawa{-}Johnson},
  editor       = {Hanseok Ko and
                  John H. L. Hansen},
  title        = {Unsupervised Text-to-Speech Synthesis by Unsupervised Automatic Speech
                  Recognition},
  booktitle    = {Interspeech 2022, 23rd Annual Conference of the International Speech
                  Communication Association, Incheon, Korea, 18-22 September 2022},
  pages        = {461--465},
  publisher    = {{ISCA}},
  year         = {2022},
  url          = {https://doi.org/10.21437/Interspeech.2022-816},
  doi          = {10.21437/INTERSPEECH.2022-816},
}

@article{Ni2025-jstti,
    author = {Junrui Ni and Liming Wang and Yang Zhang and Kaizhi Qian and Heting Gao and Mark
Hasegawa-Johnson and Chang D. Yoo},
    title = {Towards Unsupervised Speech Recognition Without Pronunciation Models},
    journal = {arXiv},
    year = 2025,
    url = {https://arxiv.org/pdf/2406.08380}
}

@inproceedings{Panayotov15-LibriSpeech,
  author    = {Vassil Panayotov and
               Guoguo Chen and
               Daniel Povey and
               Sanjeev Khudanpur},
  title     = {Librispeech: An {ASR} corpus based on public domain audio books},
  booktitle = icassp,
  pages     = {5206--5210},
  year      = {2015},
  url          = {https://doi.org/10.1109/ICASSP.2015.7178964},
  doi          = {10.1109/ICASSP.2015.7178964},
}

@inproceedings{peng2022-vghubert,
  author       = {Puyuan Peng and
                  David Harwath},
  editor       = {Hanseok Ko and
                  John H. L. Hansen},
  title        = {Word Discovery in Visually Grounded, Self-Supervised Speech Models},
  booktitle    = {Interspeech 2022, 23rd Annual Conference of the International Speech
                  Communication Association, Incheon, Korea, 18-22 September 2022},
  pages        = {2823--2827},
  publisher    = {{ISCA}},
  year         = {2022},
  url          = {https://doi.org/10.21437/Interspeech.2022-10652},
  doi          = {10.21437/INTERSPEECH.2022-10652},
}

@inproceedings{Peng2023-syllable,
  title={Syllable Segmentation and Cross-Lingual Generalization in a Visually Grounded, Self-Supervised Speech Model},
  author={Peng, Puyuan and Li, Shang-Wen and Räsänen, Okko and Mohamed, Abdelrahman and Harwath, David},
  booktitle={Interspeech},
  year={2023}
}

@inproceedings{sennrich-etal-2016-neural,
    title = "Neural Machine Translation of Rare Words with Subword Units",
    author = "Sennrich, Rico  and
      Haddow, Barry  and
      Birch, Alexandra",
    editor = "Erk, Katrin  and
      Smith, Noah A.",
    booktitle = "Proceedings of the 54th Annual Meeting of the Association for Computational Linguistics (Volume 1: Long Papers)",
    month = aug,
    year = "2016",
    address = "Berlin, Germany",
    publisher = "Association for Computational Linguistics",
    url = "https://aclanthology.org/P16-1162/",
    doi = "10.18653/v1/P16-1162",
    pages = "1715--1725"
}

@inproceedings{shi1997normalized,
  author    = {Jianbo Shi and Jitendra Malik},
  title     = {Normalized Cuts and Image Segmentation},
  booktitle = {Proceedings of the IEEE Conference on Computer Vision and Pattern Recognition (CVPR)},
  year      = {1997},
  pages     = {731--737},
  doi       = {10.1109/CVPR.1997.609407}
}

@inproceedings{shi21c_interspeech,
  title     = {AISHELL-3: A Multi-Speaker Mandarin TTS Corpus},
  author    = {Yao Shi and Hui Bu and Xin Xu and Shaoji Zhang and Ming Li},
  year      = {2021},
  booktitle = {Interspeech 2021},
  pages     = {2756--2760},
  doi       = {10.21437/Interspeech.2021-755},
  issn      = {2958-1796},
}

@inproceedings{shi2023-unsupslu,
  title     = {Bridging Speech and Text Pre-trained Models with Unsupervised ASR},
  author    = {Jiatong Shi and Chan-Jan Hsu and Holam Chung and Dongji Gao and Paola Garcia and Shinji Watanabe and Ann Lee and Hung-yi Lee},
  booktitle = {Proceedings of the 2023 IEEE International Conference on Acoustics, Speech and Signal Processing (ICASSP)},
  pages     = {1--5},
  year      = {2023},
  publisher = {IEEE},
  doi       = {10.1109/ICASSP49357.2023.10096827}
}

@article{Tseng2024-reborn,
  author       = {Liang{-}Hsuan Tseng and
                  En{-}Pei Hu and
                  David Cheng{-}Han Chiang and
                  Yuan Tseng and
                  Hung{-}yi Lee and
                  Lin{-}Shan Lee and
                  Shao{-}Hua Sun},
  title        = {{REBORN:} Reinforcement-Learned Boundary Segmentation with Iterative
                  Training for Unsupervised {ASR}},
  journal      = {CoRR},
  volume       = {abs/2402.03988},
  year         = {2024},
  url          = {https://doi.org/10.48550/arXiv.2402.03988},
  doi          = {10.48550/ARXIV.2402.03988},
  eprinttype    = {arXiv},
  eprint       = {2402.03988},
}

@inproceedings{Vaswani2017,
  title={Attention is all you need},
  author={Vaswani, Ashish and Shazeer, Noam and Parmar, Niki and Uszkoreit, Jakob and Jones, Llion and Gomez, Aidan N and Kaiser, {\L}ukasz and Polosukhin, Illia},
  booktitle=nips,
  pages={5998--6008},
  year={2017}
}

@inproceedings{wang-etal-2023-unsupasr-theory,
    title = "A Theory of Unsupervised Speech Recognition",
    author = "Wang, Liming  and
      Hasegawa-Johnson, Mark  and
      Yoo, Chang",
    booktitle = "Proceedings of the 61st Annual Meeting of the Association for Computational Linguistics (Volume 1: Long Papers)",
    month = jul,
    year = "2023",
    address = "Toronto, Canada",
    publisher = "Association for Computational Linguistics",
    pages = "1192--1215",
}

@inproceedings{wang-etal-2023-simple,
  title     = {Simple and Effective Unsupervised Speech Translation},
  author    = {Changhan Wang and Hirofumi Inaguma and Peng-Jen Chen and Ilia Kulikov and Yun Tang and Wei-Ning Hsu and Michael Auli and Juan Pino},
  booktitle = {Proceedings of the 61st Annual Meeting of the Association for Computational Linguistics (Volume 1: Long Papers)},
  pages     = {10771--10784},
  year      = {2023},
  month     = jul,
  address   = {Toronto, Canada},
  publisher = {Association for Computational Linguistics},
  doi       = {10.18653/v1/2023.acl-long.602},
  url       = {https://aclanthology.org/2023.acl-long.602/}
}

@inproceedings{wang-etal-2023-unsup-speech2sign,
author       = {Liming Wang and
                  Junrui Ni and
                  Heting Gao and
                  Jialu Li and
                  Kai Chieh Chang and
                  Xulin Fan and
                  Junkai Wu and
                  Mark Hasegawa{-}Johnson and
                  Chang Dong Yoo},
  editor       = {Anna Rogers and
                  Jordan L. Boyd{-}Graber and
                  Naoaki Okazaki},
  title        = {Listen, Decipher and Sign: Toward Unsupervised Speech-to-Sign Language
                  Recognition},
  booktitle    = {Findings of the Association for Computational Linguistics: {ACL} 2023,
                  Toronto, Canada, July 9-14, 2023},
  pages        = {6785--6800},
  publisher    = {Association for Computational Linguistics},
  year         = {2023},
  url          = {https://doi.org/10.18653/v1/2023.findings-acl.424},
  doi          = {10.18653/V1/2023.FINDINGS-ACL.424},
}

@inproceedings{wang2024unsupervised,
  author       = {Liming Wang and
                  Mark Hasegawa{-}Johnson and
                  Chang D. Yoo},
  title        = {Unsupervised Speech Recognition with N-skipgram and Positional Unigram
                  Matching},
  booktitle    = {{IEEE} International Conference on Acoustics, Speech and Signal Processing,
                  {ICASSP} 2024, Seoul, Republic of Korea, April 14-19, 2024},
  pages        = {10936--10940},
  publisher    = {{IEEE}},
  year         = {2024},
  url          = {https://doi.org/10.1109/ICASSP48485.2024.10446327},
  doi          = {10.1109/ICASSP48485.2024.10446327},
}

@article{chu2024qwen2,
  title={Qwen2-audio technical report},
  author={Chu, Yunfei and Xu, Jin and Yang, Qian and Wei, Haojie and Wei, Xipin and Guo, Zhifang and Leng, Yichong and Lv, Yuanjun and He, Jinzheng and Lin, Junyang and others},
  journal={arXiv preprint arXiv:2407.10759},
  year={2024}
}

@article{defossez2024moshi,
  title={Moshi: a speech-text foundation model for real-time dialogue},
  author={D{\'e}fossez, Alexandre and Mazar{\'e}, Laurent and Orsini, Manu and Royer, Am{\'e}lie and P{\'e}rez, Patrick and J{\'e}gou, Herv{\'e} and Grave, Edouard and Zeghidour, Neil},
  journal={arXiv preprint arXiv:2410.00037},
  year={2024}
}

@inproceedings{conneau21_interspeech,
  title     = {Unsupervised Cross-Lingual Representation Learning for Speech Recognition},
  author    = {Alexis Conneau and Alexei Baevski and Ronan Collobert and Abdelrahman Mohamed and Michael Auli},
  year      = {2021},
  booktitle = {Interspeech 2021},
  pages     = {2426--2430},
  doi       = {10.21437/Interspeech.2021-329},
  issn      = {2958-1796},
}

@book{cover2006-it,
    author = {T. M. Cover and J. A. Thomas},
    title = {Elements of Information Theory},
    edition   = {2},
    year      = {2006},
    publisher = {Wiley-Interscience},
}

@inproceedings{hoshen2018non,
  title={Non-Adversarial Unsupervised Word Translation},
  author={Hoshen, Yedid and Wolf, Lior},
  booktitle={Proceedings of the 2018 Conference on Empirical Methods in Natural Language Processing},
  pages={469--478},
  year={2018}
}

@inproceedings{artetxe2018unsupervised,
  title={Unsupervised Statistical Machine Translation},
  author={Artetxe, Mikel and Labaka, Gorka and Agirre, Eneko},
  booktitle={Proceedings of the 2018 Conference on Empirical Methods in Natural Language Processing},
  pages={3632--3642},
  year={2018}
}

@inproceedings{artetxe-etal-2019-effective,
    title = "An Effective Approach to Unsupervised Machine Translation",
    author = "Artetxe, Mikel  and
      Labaka, Gorka  and
      Agirre, Eneko",
    editor = "Korhonen, Anna  and
      Traum, David  and
      M{\`a}rquez, Llu{\'i}s",
    booktitle = "Proceedings of the 57th Annual Meeting of the Association for Computational Linguistics",
    month = jul,
    year = "2019",
    address = "Florence, Italy",
    publisher = "Association for Computational Linguistics",
    url = "https://aclanthology.org/P19-1019/",
    doi = "10.18653/v1/P19-1019",
    pages = "194--203",
}

@inproceedings{lample2018phrase,
  title={Phrase-Based \& Neural Unsupervised Machine Translation},
  author={Lample, Guillaume and Ott, Myle and Conneau, Alexis and Denoyer, Ludovic and Ranzato, Marc’Aurelio},
  booktitle={Proceedings of the 2018 Conference on Empirical Methods in Natural Language Processing},
  pages={5039--5049},
  year={2018}
}

@article{Pratap2020MLSAL,
  title={MLS: A Large-Scale Multilingual Dataset for Speech Research},
  author={Vineel Pratap and Qiantong Xu and Anuroop Sriram and Gabriel Synnaeve and Ronan Collobert},
  journal={ArXiv},
  year={2020},
  volume={abs/2012.03411}
}

@INPROCEEDINGS{Zhang2009-syllable,
  author={Yaodong Zhang and James R. Glass},
  booktitle=icassp, 
  title={Speech rhythm guided syllable nuclei detection}, 
  year={2009},
  volume={},
  number={},
  pages={3797-3800},
  doi={10.1109/ICASSP.2009.4960454}}

@inproceedings{yu2022automatic,
    title={Automatic Speech Recognition Datasets in Cantonese: A Survey and New Dataset},
    booktitle={Proceedings of the Thirteenth Language Resources and Evaluation Conference},
    author={Tiezheng Yu and Rita Frieske and Peng Xu and Samuel Cahyawijaya and Cheuk Tung Shadow Yiu and Holy Lovenia and Wenliang Dai and Elham J. Barezi and Qifeng Chen and Xiaojuan Ma and Bertram E. Shi and Pascale Fung},
    year={2022},
    url={https://aclanthology.org/2022.lrec-1.696/},
}

@misc{ithuan2019suisiann,
  author       = {{Ì-thuân}},
  title        = {{Suísiann Dataset}},
  year         = {2019},
  howpublished = {\url{https://suisiann-dataset.ithuan.tw/}},
  note         = {Accessed: 2026-01-22}
}

@inproceedings{Drinea2007,
  author={Drinea, Eleni and Kirsch, Adam},
  booktitle={2007 IEEE International Symposium on Information Theory}, 
  title={Directly Lower Bounding the Information Capacity for Channels with I.I.D. Deletions and Duplications}, 
  year={2007},
  volume={},
  number={},
  pages={1731-1735},
  doi={10.1109/ISIT.2007.4557471},
}

@book{KandybowiczTorrence2017,
  editor    = {Jason Kandybowicz and Harold Torrence},
  title     = {Africa's Endangered Languages: Documentary and Theoretical Approaches},
  publisher = {Oxford University Press},
  year      = {2017}
}

@article{FarukhVulchanova2014,
  author  = {A. Farukh and Mila Vulchanova},
  title   = {Predictors of Reading in Urdu: Does Deep Orthography Have an Impact?},
  journal = {Dyslexia},
  volume  = {20},
  number  = {2},
  pages   = {146--166},
  year    = {2014}
}

@misc{KhanEtAl2024PashtoNLP,
  author       = {Zohaib Ahmad Khan and Yuanqing Xia and Fiza Khaliq and Javed Ali Khan and Nek Dil Khan},
  title        = {From Tradition to Technology: A Systematic Survey on Navigating Pashto in Modern {NLP}},
  year         = {2024},
  howpublished = {SSRN},
  doi          = {10.2139/ssrn.5031721},
  note         = {Available at SSRN}
}
